\documentclass[10pt,twocolumn]{IEEEtran}
\usepackage[latin1]{inputenc}
\usepackage[T1]{fontenc}
\usepackage{algorithm,algorithmic}
\usepackage{amsmath,amsthm,amssymb}

\usepackage{graphicx}

\usepackage{color}
\usepackage[x11names]{xcolor}
\usepackage{cite}
\usepackage{enumerate}

\usepackage{color}
\usepackage{hyperref}
\usepackage{psfrag}
\makeatletter
\newcommand{\BIG}{\bBigg@{3}}
\newcommand{\BIGG}{\bBigg@{4}}
\makeatother
\usepackage{color}

\newtheorem{proposition}{Proposition}

\newtheorem{remark}{Remark}
\newtheorem{corollary}{Corollary}

\begin{document}
	
	\title{On the Secrecy Performance and Power Allocation in Relaying Networks with Untrusted Relay in the Partial Secrecy Regime}
	\author{\normalsize Diana Pamela Moya~Osorio,~\IEEEmembership{\normalsize Member,~IEEE},~Hirley~Alves,~\IEEEmembership{\normalsize Member,~IEEE}, and Edgar~Eduardo~Benitez~Olivo,~\IEEEmembership{\normalsize Member,~IEEE}
		\thanks{D.~P.~Moya~Osorio and H.~Alves are with the Centre for Wireless Communications (CWC), University of Oulu, Finland (e-mail: \{diana.moyaosorio; hirley.alves\}@oulu.fi).}
		\thanks{D.~P.~Moya~Osorio is also with the Department of Electrical Engineering, Center of Exact Sciences and Technology, Federal University of S\~{a}o Carlos, 13565-905 S\~{a}o Carlos, Brazil (e-mail: dianamoya@ufscar.br).}
		\thanks{ E.~E.~Benitez~Olivo is with S\~{a}o Paulo State University (UNESP), Campus of S\~{a}o Jo\~{a}o da Boa Vista, 13876-750 S\~{a}o Jo\~{a}o da Boa Vista, Brazil (e-mail: edgar.olivo@unesp.br).}
		\thanks{This work was supported in part by S\~{a}o Paulo Research Foundation (FAPESP) under Grant 2017/20990-6, in part by the Academy of Finland 6Genesis Flagship under Grant 318927, in part by the EE-IoT project under Grant 319008, and in part by the Brazilian National Council for Scientific and Technological Development (CNPq) under Grants 428649/2016-5 and 421850/2018-3.}
	}
	
	
	\maketitle
	
	\begin{abstract}
		Recently, three useful secrecy metrics based on the partial secrecy regime were proposed to analyze secure transmissions on wireless systems over quasi-static fading channels, namely: generalized secrecy outage probability, average fractional equivocation, and average information leakage. These metrics were devised from the concept of fractional equivocation, which is related to the decoding error probability at the eavesdropper, so as to provide a comprehensive insight on the practical implementation of wireless systems with different levels of secrecy requirements. Considering the partial secrecy regime, in this paper we examine the secrecy performance of an amplify-and-forward relaying network with an untrusted relay node, where a destination-based jamming is employed to enable secure transmissions. In this regard, a closed-form approximation is derived for the generalized secrecy outage probability, and integral-form expressions are obtained for the average fractional equivocation and the average information leakage rate. Additionally, equal and optimal power allocation schemes are investigated and compared for the three metrics. From this analysis, we show  that different power allocation approaches lead to different system design criteria. The obtained expressions are validated via Monte Carlo simulations.
	\end{abstract}
	
	
	\section{Introduction}
	\label{sec:intro}
	Novel applications and services are envisioned with the implementation of the fifth generation (5G) of wireless networks, which are highly demanding in  terms  of  reliability, latency,  energy efficiency, spectrum  efficiency,   flexibility,  and  connection  density. Thus, the International Telecommunications Union (ITU) has categorized 5G services into three broad groups: enhanced mobile broadband (eMBB), massive machine-type communication (mMTC), and ultra-reliable and low-latency communication (URLLC)~\cite{art:3gpp}, which will impulse the progressive implementation of the Internet of Things (IoT) paradigm. Considering these three service categories, it can be foreseen that the upcoming applications for 5G and beyond are extremely vulnerable to security breaches.
	
	Traditionally, network security is provided by bit-level cryptographic techniques and  the corresponding  protocols  at the different layers of the  data  processing  stack. However, these solutions  present some weaknesses regarding public  wireless  networks involving  restrictions and high costs for the users. Therefore, traditional approaches are not sufficient for guaranteeing confidentiality in 5G networks and beyond~\cite{art:wu2018}. A new paradigm for providing enhanced security in wireless networks is referred to as physical layer security (PLS), which can potentially offer secure transmissions by efficiently exploiting the properties of wireless medium (fading, interference, and diversity)\cite{art:poor2017}. 
	
	In 1949, Shannon introduced the concepts on secrecy transmissions from the information theoretic perspective, when he proposed the so-called cypher system in his pioneering work in~\cite{art:shannon1949}. In that system, it is considered a noiseless channel where a transmitter (Alice) intends to communicate with a legitimate receiver (Bob), by sharing a secret key $K$ that Alice uses to encrypt a message $M$ into a codeword $X$, in order to  maintain this in secret from an eavesdropper (Eve) that intercept the message. Therein, it was defined the perfect secrecy as the condition of $X$ revealing no information about $M$, i.e., the mutual information $I(M;X)=0$, so that $M$ and $X$ must be statistically independent, thus the best that Eve can do is to guess the transmitted message.
	
	Later, in 1975, Wyner showed that secrecy can be attained by exploiting the qualities of the channels without the need of a shared key~\cite{art:wyner}. In that work, it was proposed the discrete memoryless wiretap channel, where Alice must encode $M$ into a $n$-length codeword $X^n$, while $Y^n$ and $Z^n$ are the outputs at Bob's and Eve's channels, respectively. Therein, Wyner defined the concept of weak secrecy, which establishes that the statistical independence between $M$ and $Z^n$ is only required asymptotically in the block length $n$, i.e., when $\lim_{n\rightarrow\infty} ({1}/{n}) I(M;Z^n)$$=$$0$, thus describing the largest rate at which Eve gains no information about $M$ by observing $Z^n$. That definition was strengthen in~\cite{art:maurer2000} by disregarding the term $1/n$ in the definition of weak secrecy, thus defining the concept of strong secrecy, and the intuition is that the information leaked to Eve vanishes as $n$ approaches $\infty$. Besides, in~\cite{art:wyner}, the so-called secrecy capacity was defined as the maximal rate at which both reliability and security can be achieved, thus characterizing a rate-equivocation region, where the rate
	corresponds to the rate of reliable communication between Alice and Bob, and the equivocation defines the uncertainty about the message received by Eve. However, that definition requires that the Bob's channel must be less noisy than Eve's channel, thus only working for discrete memoryless channels. Later, in~\cite{art:csiszar}, Wyner's results were generalized for the non-degraded case.   
	Moreover, in~\cite{art:leung}, the secrecy capacity for the Gaussian wiretap channel was studied, wherein it was established that the secrecy capacity is the difference between the capacities of the legitimate and eavesdropper channels; therefore, a secure communication is possible if and only if the signal-to-noise ratio (SNR) of  the legitimate channel  is larger than that of the eavesdropper channel. 
	
	Furthermore, in~\cite{art:bloch2008} fading wiretap channels were investigated in terms of ergodic secrecy capacity (ESC) and secrecy outage probability (SOP), where it was demonstrated that fading can actually be beneficial for transmitting information in a secure manner. Besides, the multiple-input-multiple-output (MIMO) wiretap channel was investigated in~\cite{art:oggier}. Then, by virtue of the great advantages that multiple antennas can offer in achieving enhanced secrecy performances, an extensive number of works have investigated the secrecy performance of MIMO  networks~\cite{art:alves2012}. On the other hand, cooperative relaying, widely recognized for offering significant gains on reliability and coverage in wireless networks, has also been considered as a promising technique to exploit the physical characteristics of wireless channels in order to improve the security of a network against eavesdropping~\cite{art:dong2010}. For instance, in~\cite{art:dong2010}, widely-known relaying protocols, namely amplify-and-forward (AF) and decode-and-forward (DF), were evaluated for cooperative networks with multiple relays by considering the secrecy rate maximization problem and power allocation subject to a power constraint, and the transmit power minimization problem subject to a secrecy rate. Additionally, in that work, it was also proposed the so-called cooperative jamming (CJ) technique, in which the relays contribute to provide secrecy by sending a jamming signal in order to interfere the eavesdroppers, thus preventing them from extracting information from the confidential message. 
	
	
	However, the reported benefits in the aforementioned works are based on the premise that the relay is a trustworthy node. Yet, in many networks, not all nodes have the same level of security clearance, such that the information must be kept confidential even from the relays. Those scenarios have raised the interest on determining whether cooperation is beneficial or not in the case of untrusted relays~\cite{art:he2010,art:Jeong2012,art:Mo2014,art:he2009,art:Sun2012,art:Wang2014,art:Kuhestani2016,art:Kuhestani2018,art:osorio2018,art:lv2017,art:Mamaghani2017,art:Kuhestani20182,art:Xu2018}. For instance, in~\cite{art:he2010}, the achievable secrecy rate for the compress-and-forward protocol was studied for two cases: ($i$) when the first-hop relaying link channel is orthogonal to the multiple access channel from the source and relay to the destination, and ($ii$) when the second-hop relaying link channel is orthogonal to the broadcast channel from the source to the relay and destination. For the first case, it was found that the secrecy capacity is attained by restricting the confidential transmission to the direct link, then the untrusted relay is not useful. However, for the second case, it was found that the secrecy rate can be improved by relying on the untrusted relay to cooperate with information transmission rather than only considering it as an eavesdropper. { Besides, in~\cite{art:Jeong2012,art:Mo2014}, a MIMO cooperative relay network was addressed, where transmit beamforming was employed both at the source and at the untrusted relay. In~\cite{art:Jeong2012}, the noncooperative secure beamforming and cooperative secure beamforming were considered, and the conditions under which the cooperative scheme achieves a higher secrecy rate than the noncooperative scheme were characterized at the high SNR regime. In~\cite{art:Mo2014}, the jointly optimization of beamformers at the source and relay was performed in order to maximize the secrecy sum rate of a two-way communication; then, it was shown the advantages of the signal alignment against eavesdropping}. Further, in~\cite{art:he2009}, a positive secrecy rate was obtained by relying on the destination node or an external node to send a jamming signal in a two-hop compress-and forward relaying network. This technique is referred to as destination-based jamming (DBJ). { In~\cite{art:Sun2012}, the ESC was investigated for the DBJ scheme, by considering the single-relay and multiple-relay scenarios. For the latter, a secure relay selection was proposed and the results showed that, from a secrecy perspective, the system performance worsens as the number of relays increases}. 
	{  Also, DBJ was explored in~\cite{art:Wang2014,art:Kuhestani2016,art:Kuhestani2018} along with optimal power allocation based on instantaneous channel estimations, where the impact of large scale antenna arrays at either the source or the destination was examined for the ESC in~\cite{art:Wang2014} and for the SOP in~\cite{art:Kuhestani2016}, while hardware impairments were also considered in~\cite{art:Kuhestani2018}}. 
	
	In~\cite{art:osorio2018}, the impact of the direct link on the secrecy outage probability was analyzed for a relaying network with multiple untrusted AF relays, where partial relay selection and DBJ are considered by means of a full-duplex destination. 
	{ Moreover, the ergodic secrecy sum rate was evaluated in~\cite{art:Mamaghani2017} by considering a two-way communication with a friendly jammer and an energy-constrained untrusted relay. In that work, wireless energy transfer technique was employed to charge the relay, and the impact of the time expended to charge the relay was evaluated. In~\cite{art:Kuhestani20182}, the authors focused on the analysis of the security-reliability trade-off based on the connection outage probability and the intercept probability under hardware impairments, for the case of direct transmission without using the relay, and the case of destination-based cooperative jamming. Moreover, in~\cite{art:Xu2018}, a technique called constellation overlapping was proposed to provide security to two-way untrusted relaying systems in a resource-efficient manner. For the symmetric case, where two terminal users adopt the same modulation type, a full constellation overlapping was obtained, and an error floor is caused at the untrusted relay. For the asymmetric case, a constellation virtualization method was employed in order to achieve the constellation overlapping effect.}
	
	Nevertheless, all the aforementioned works above are mainly based on the classical information-theoretic definition of secrecy for quasi-static fading channels, whereby the performance metrics are established from the premise that the eavesdropper's decoding error probability is equal to 1. Hence, the classical SOP shows to be limited for an appropriate design of practical secure wireless communication systems, since it establishes an extremely stringent assumption, thus presenting three important shortcomings. First, the SOP does not allow to obtain appropriate insights on the eavesdropper's ability to decode confidential messages. Second, this metric cannot characterize the amount of information that is leaked to the eavesdropper when an outage occurs. Third, it cannot be associated with the quality of service (QoS) demands of different applications and services. In this regard, recently in~\cite{art:he}, new secrecy metrics for wireless transmissions focusing on quasi-static fading channels were proposed, namely: generalized secrecy outage probability (GSOP), average fractional equivocation (AFE), and average information leakage rate (AILR). These metrics are based on the so-called partial secrecy  regime, whereby a  system  is evaluated  by means of the fractional equivocation, which  regards to  the level  at  which  the  eavesdropper  is  confused.
	
	In light of the above considerations, this paper contributes to further extend the understanding of the secrecy performance of cooperative communications with untrusted relays. To this end, we examine a three-node AF relaying network where the relay is untrusted and a DBJ protocol is used to enable secure transmissions. Differently from previous related works, we analyze the secrecy performance in terms of the metrics proposed in~\cite{art:he}, for which we also consider different power allocation policies among the source, relay, and destination. The following are the main contributions of this paper:
	\begin{itemize}
		\item A closed-form approximate expression is derived for the GSOP, which allows to associate the concept of secrecy outage with the ability of the untrusted relay to decode confidential information. Moreover, analytical expressions in a one-fold integral form are obtained for the AFE and AILR. The former is an asymptotic lower bound on the untrusted relay's decoding error probability, and the latter describes how much and how fast the information is leaked to the untrusted relay. 
		\item For the aforementioned metrics, simple closed-form asymptotic expressions are obtained joint with the system diversity order.
		\item For the three aforementioned performance metrics, equal power allocation (EPA) and optimal power allocation (OPA) schemes are compared.  
	\end{itemize}

	Throughout this paper, $f_Z \left( \cdot \right)$ and $F_Z \left( \cdot \right)$ denote the probability density function (PDF) and the cumulative distribution function (CDF) of a random variable $Z$, respectively, $\mathcal{E}\left\{ \cdot \right\}$ denotes expectation, $\Pr\left(\cdot\right)$ denotes probability, and $W\left[x\right]$ is the principal value of the Lambert-$W$ function~\cite{art:lambert}.
	\section{System Model}
	
	\label{sec:system_model}
	\begin{figure}\centering
		\psfrag{S}[Bc][Bc][1]{S}
		\psfrag{D}[Bc][Bc][1]{D}
		\psfrag{h1}[Bc][Bc][0.8]{$h_{\mathrm{SR}}$}
		\psfrag{h2}[Bc][Bc][0.8]{$h_{\mathrm{RD}}$}
		
		\psfrag{h0}[Bc][Bc][0.7]{$h_{\mathrm{SD}}$}
		\psfrag{h4}[Bc][Bc][0.7]{$h_{\mathrm{DD}}$}
		\psfrag{f1}[Bl][Bl][0.82]{first phase}
		\psfrag{fj}[Bl][Bl][0.82]{first phase: jamming}
		\psfrag{f2}[Bl][Bl][0.82]{second phase}
		\psfrag{R}[Bc][Bc][1]{R}
		\psfrag{Primeira Fase}[Bc][Bc][0.6]{First Phase}
		\psfrag{Segunda Fase}[Bc][Bc][0.6]{Second Phase}
		\includegraphics[width=0.8\linewidth]{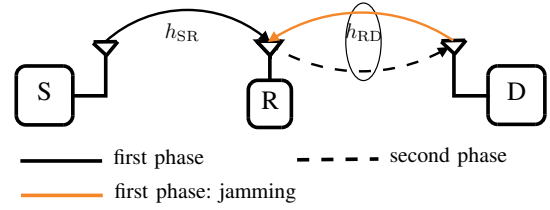}
		\caption{System Model of a three-node relaying network with an untrusted relay and DBJ.}
		\label{fig:systemodel}
	\end{figure}
	
	The system model depicted in Fig.~\ref{fig:systemodel} illustrates a cooperative relaying network consisting of a single-antenna source S that tries to communicate with a single-antenna destination D with the help of a single-antenna AF relay R. The relay node is considered to have a lower level of security clearance to access the confidential information transmitted by S, then it is an untrusted node that can potentially attempt to decode the message from the signal sending by S. In this system, the direct link is considered to be extremely attenuated, so that it is negligible, { and the widely-used wiretap codes are assumed for message transmissions~\cite{art:wyner}.  Thus, the codebook is assumed to be defined by two rate parameters, the codeword transmission rate $R_T$$=$$\tfrac{H\left(X^n\right)}{n}$ and the confidential information rate or target secrecy rate $R_S$$=$$\tfrac{H\left(M\right)}{n}$, where $R_S \leq R_T$. Hence, the wiretap code is constructed by generating $2^{nR_T}$ codewords $x^n\left(u,v\right)$, where $u$$=$$\{1,2,...,2^{nR_S}\}$ and $v$$=$$\{1,2,...,2^{n\left(R_T-R_S\right)}\}$. Then, for each message of index $u$, it is randomly selected an index $v$ with uniform distribution, such that the codeword $x^n\left(u,v\right)$ is transmitted.} The communication process is performed in two phases, as described below:
	
	\begin{itemize}
		\item \textit{First phase:} in this phase, S sends an information signal $s_I(t)$ to R; meanwhile, D sends an artificial jamming signal $s_J(t)$ to hinder the relay from eavesdropping the information sent by S.
		\item \textit{Second phase:} in this phase, the relay amplifies the signal received from S, which was interfered by the jamming signal sent by D, and forwards it to the destination. Since D knows the jamming signal transmitted in the previous phase, this can be subtracted f rom the received signal.
	\end{itemize}
	
	All links in this network are considered to undergo quasi-static Rayleigh fading, as well as additive white Gaussian noise (AWGN) with mean power $N_0$. Therefore, the channel coefficients for the links S$\rightarrow$R,  R$\rightarrow$D and D$\rightarrow$R, denoted by $h_{\mathrm{SR}}$, $h_{\mathrm{R}\mathrm{D}}$, and $h_{\mathrm{DR}}$, respectively, are independent complex circularly-symmetric Gaussian random variables with variance $\Omega_{\mathrm{AB}}\mathop{=} E\{|h_{\mathrm{AB}}|^2\}$, that is $\mathcal{CN}\left(0,\Omega_{\mathrm{AB}}\right)$, with $\mathrm{A}\mathop{\in}\{\mathrm{S}, \mathrm{R}, \mathrm{D}\}$ and $\mathrm{B}\mathop{\in}\{\mathrm{R},\mathrm{D}\}$. Accordingly, the channel gains $g_{\mathrm{AB}}=|h_{\mathrm{AB}}|^2$ are exponentially distributed with mean value $\Omega_{\mathrm{AB}}$. 
	Thus, the instantaneous received signal-to-noise ratios (SNRs) at the first-hop relaying link, second-hop relaying link, and jamming link are, respectively, given by $\gamma_{\mathrm{SR}}$$=$$ g_{\mathrm{SR}}P_\mathrm{S}/N_0$, $\gamma_{\mathrm{RD}}$$=$$ g_{\mathrm{R}\mathrm{D}}P_\mathrm{R}/N_0$, and $\gamma_{\mathrm{DR}}$$=$$ g_{\mathrm{DR}}P_\mathrm{D}/N_0$, where $P_\mathrm{S}$, $P_\mathrm{R}$, and $P_\mathrm{D}$ are the transmit powers at source, relay and destination. In addition, the total transmit system power is assumed to be limited to a value of $P$, for the whole transmission process. Accordingly, by denoting the total transmit system SNR as $\gamma_P=P/N_0$, the transmit SNRs at S, D, and R can be respectively written as $\gamma_{{\mathrm{S}}} \mathop{=} P_{\mathrm{S}}/N_0 \mathop{=}\eta_1 \gamma_P$, $\gamma_{{\mathrm{R}}}\mathop{=}{P_\mathrm{R}}/N_0\mathop{=}\eta_2 \gamma_P$, and $\gamma_{{\mathrm{D}}} \mathop{=}{P_\mathrm{D}}/N_0\mathop{=}\eta_3 \gamma_P$, where $\eta_1$, $\eta_2$, and $\eta_3$ are the power allocation factors, with $\eta_1+\eta_2+\eta_3=1$.
	
	Under the above assumptions, the received signal at R during the first phase is given by
	\begin{align}
	y_{\mathrm{R}}\left(t\right)=&\sqrt{P_\mathrm{S}} h_{\mathrm{SR}}s_I\left(t\right) {+} \sqrt{P_\mathrm{D}} h_{\mathrm{DR}} s_J\left(t\right) {+} n_{\mathrm{R}}\!\left(t\right),\label{eq:rxsignalR}
	\end{align}
	where the mean power of the signals $s_I(t)$ and $s_J(t)$ are normalized to unity, that is  $E\{|s_I\left(t\right)|^2\}=E\{|s_J\left(t\right)|^2\}=1$, and $n_{\mathrm{R}}\!\left(t\right)$ is the AWGN component at R. 
	
	On the other hand, the received signal at D during the second phase is given by
	\begin{align}
	y_{\mathrm{D}}\left(t\right)=&\sqrt{P_\mathrm{R} } h_{\mathrm{R}\mathrm{D}} \mathcal{G} y_{\mathrm{R}}\left(t\right) + n_{\mathrm{D}}\!\left(t\right),\label{eq:rxsignalD2}
	\end{align}
	where $n_{\mathrm{D}}\!\left(t\right)$ is the AWGN component at D, and $\mathcal{G}$ is the amplification factor related to the AF protocol, given as
	\begin{align}\label{eq:ampfactor}
	\mathcal{G}=\dfrac{1}{\sqrt{P_\mathrm{S}g_{\mathrm{SR}}+P_\mathrm{D}g_{\mathrm DR}+N_0}}.
	\end{align}
	Hence, by substituting~\eqref{eq:rxsignalR} into~\eqref{eq:rxsignalD2} and considering that the jamming signal can be effectively removed at D, the received signal at D during the second phase can be rewritten as
	\begin{align}
	y_{\mathrm{D}}\left(t\right)\!=\!\sqrt{P_\mathrm{R} } h_{\mathrm{R}\mathrm{D}} \mathcal{G} \sqrt{P_\mathrm{S}} h_{\mathrm{SR}}s_I\left(t\right) \!+\! \sqrt{P_\mathrm{R} } h_{\mathrm{R}\mathrm{D}} \mathcal{G} n_{\mathrm{R}}\!\left(t\right) \!+\! n_{\mathrm{D}}\!\left(t\right).\label{eq:rxsignalD2r}
	\end{align}
	Thus, after some mathematical manipulation, the instantaneous signal-to-interference-plus-noise ratio (SINR) received at R and D, during the first and second phase, can respectively be expressed as
	\begin{align}
	\Gamma_{\mathrm{R}}=& \dfrac{{P_\mathrm{S}} g_{\mathrm{SR}}}{{P_\mathrm{D}} g_{\mathrm{R}\mathrm{D}} +N_0}=\dfrac{\gamma_{\mathrm{SR}}}{\gamma_{\mathrm{RD}}+1},\label{eq:sinrR}\\
	\Gamma_{\mathrm{D}} = & \dfrac{{P_\mathrm{S}} g_{\mathrm{SR} }{P_\mathrm{R}} g_{\mathrm{R}\mathrm{D}}}{{P_\mathrm{S}} g_{\mathrm{SR} }N_0+{P_\mathrm{R}} g_{\mathrm{R}\mathrm{D}}N_0+{P_\mathrm{D}} g_{\mathrm{R}\mathrm{D}}N_0+N_0^2}\nonumber\\
	= &\dfrac{\gamma_{\mathrm{SR}}\gamma_{\mathrm{RD}}}{\gamma_{\mathrm{SR}}+\gamma_{\mathrm{RD}}+\gamma_{\mathrm{DR}}+1},\label{eq:snrD2}
	\end{align}
	where we have assumed reciprocity between the R$\rightarrow$D and 
	D$\rightarrow$R links, so that in the following we consider $\gamma_{\mathrm{RD}}=\gamma_{\mathrm{DR}}$.
	
	\section{Secrecy Performance Analysis}
	\label{sec:outage_analysis}
	
	The partial secrecy can be quantified by the fractional equivocation, which is an asymptotic lower bound of the decoding error probability at the eavesdropper; therefore, it is related to the capability of the eavesdropper on decoding the confidential message, being defined as~\cite{art:leung}
	\begin{align}\label{eq:Delta}
	\Delta=\dfrac{H\left(M|Z^n\right)}{H\left(M\right)}
	\end{align}
	where $M$ is the confidential information, $Z^n$ is the received information at the eavesdropper, and $H\left(M\right)$ and $H\left(M|Z^n\right)$  are the entropy of the source's information and the residual uncertainty for the message at the eavesdropper, respectively. From this, the fractional equivocation for a given fading realization of the wireless channel is given by~\cite{art:he}
	\begin{align}\label{eq:Deltau_1b}
	\Delta=\left\lbrace 
	\begin{array}{ll}
	1, & \mathrm{if}\ C_E \leq C_L-R_S \\ 
	\left(C_L-C_E\right)/R_S, & \mathrm{if}\ C_L-R_S< C_E < C_L\\
	0, & \mathrm{if}\ C_L\leq C_E,
	\end{array}  \right.
	\end{align}
	where, for the proposed system, $C_{L}=\tfrac{1}{2}\log_2 \left(1+\Gamma_\mathrm{D}\right)$ is the capacity of the legitimate channel; $C_{E}=\tfrac{1}{2}\log_2 \left(1+\Gamma_\mathrm{R}\right)$ is the capacity of the eavesdropper channel (at the S$\rightarrow$R link); and $R_S$ is the secrecy rate. Thus, by defining the random variable { $\Phi \buildrel \Delta \over =\tfrac{1+\Gamma_\mathrm{D}}{1+\Gamma_\mathrm{R}}$}, the fractional equivocation in~\eqref{eq:Deltau_1b} can be rewritten~as
	\begin{align}\label{eq:Deltau_1br}
	\Delta=\left\lbrace 
	\begin{array}{ll}
	1, & \mathrm{if}\ \Phi \geq 2^{2 R_S} \\ 
	\tfrac{1}{2 R_S} \log_2 \Phi, & \mathrm{if}\ \Phi < 2^{2 R_S} < \Phi 2^{2 R_S}\\
	0, & \mathrm{if}\ \Phi \leq 1.
	\end{array}  \right.
	\end{align}
	
	A characterization of $\Delta$ measures the long-term performance of a system with time-varying channel realizations, from which the GSOP, AFE, and AILR can be investigated~\cite{art:leung}. In the following, these three secrecy metrics are analyzed for the proposed system. 
	
	\subsection{Generalized Secrecy Outage Probability -- GSOP}
	\label{subsec:gsop}
	
	The GSOP characterizes the probability that the fractional equivocation is lower than a certain value $\theta$ or, in other words, the information leakage ratio, $1-\Delta$, is larger than a certain value, $1-\theta$, and it can be expressed as~\cite{art:he}
	\begin{align}\label{eq:gsop}
	\mathrm{GSOP}= \Pr\left(\Delta<\theta\right),
	\end{align}
	where $0<\theta\leq1$ denotes the minimum acceptable value of the fractional equivocation, and it is related to the ability of R on extracting information from the confidential message sent by S. An approximate expression for the GSOP of the proposed system, which is highly accurate for the medium-to-high SNR regime, can be obtained as stated in the following proposition.
	
	\begin{proposition}\label{prop:gsop}
		By defining $\tau_1 \buildrel \Delta \over = 2^{2 R_S \theta}$, an approximation for the generalized secrecy outage probability of a three-node AF relaying network with an untrusted relay is given by
		\begin{align}\label{eq:gsopfinal}
		\nonumber \mathrm{GSOP}\approx&\sqrt{\pi } \eta_1\tau_1  \Omega_{\mathrm{SR}}\psi_3 e^{\psi_6} \mathrm{erfc}\left[\psi_3 (\eta_1\tau_1  \Omega_{\mathrm{SR}}+\Omega_{\mathrm{RD}} \sqrt{\psi_2})\right]\\ &{-}e^{-\psi_5} {+}\left(e^{-\psi_5}{-}e^{-\psi_4}\right) e^{-\frac{\eta_1\Omega_{\mathrm{SR}}\psi_5}{\Omega_{\mathrm{RD}} (\eta_2+\eta_3)}}-e^{\psi_7}+e^{-\psi_8}{+}1,
		\end{align}
		where,
		\begin{align*}
		\nonumber	\psi_1=&\sqrt{\frac{(\eta_2\!+\!\eta_3)^2\! \left(\!\eta_2^2 (\tau_1\! -\!1)^2\!+\!2 \eta_2 \eta_3 \!\left(2 \tau_1 ^2\!-\!\tau_1\! -\!1\right)\!+\!\eta_3^2 (1\!-\!2 \tau_1 )^2\right)}{\gamma_P^2 \eta_1^2 \eta_2^2 \eta_3^2}},\\ 
		\nonumber \psi_2=&\eta_2^2 \left(2 \tau_1 ^2-2 \tau_1 +1\right)+\eta_3^2 \left(5 \tau_1 ^2-4 \tau_1 +1\right)\\
		\nonumber &+2 \eta_2 \eta_3 \left(3 \tau_1 ^2-\tau_1 +\gamma_P \eta_1 \tau_1  \psi_1-1\right),\\ 
		\nonumber	\psi_3=&\frac{1}{2 \Omega_{\mathrm{RD}} \sqrt{\gamma_P \eta_1 \eta_2 \eta_3 \tau_1  \Omega_{\mathrm{SR}}}},\\
		\nonumber	\psi_4=&\frac{(\tau_1 -1) (\eta_2+\eta_3)}{\gamma_P \eta_1 \eta_2 \Omega_{\mathrm{SR}}},\\ 
		\nonumber \psi_5=&\frac{(\eta_2\!+\!\eta_3) \!\sqrt{\!-2 \tau_1  \left(\eta_2^2\!+\!\eta_2 \eta_3\!+\!2 \eta_3^2\right)\!+\!\tau_1 ^2 (\eta_2\!+\!2 \eta_3)^2\!+\!(\eta_2\!-\!\eta_3)^2}}{2 \gamma_P \eta_1 \eta_2 \eta_3 \Omega_{\mathrm{SR}}}\\
		\nonumber &+\frac{(\eta_2+\eta_3) \left(\eta_2 (\tau_1 -1)+\eta_3 (2 \tau_1 -1)\right)}{2 \gamma_P \eta_1 \eta_2 \eta_3 \Omega_{\mathrm{SR}}},\\
		\nonumber \psi_6=&\frac{2 \eta_2 \Omega_{\mathrm{RD}} (\eta_1 \tau_1  \Omega_{\mathrm{SR}}+\eta_3 (\tau -1) \Omega_{\mathrm{RD}})+\eta_2^2 \Omega_{\mathrm{RD}}^2}{4 \gamma_P \eta_1 \eta_2 \eta_3 \tau_1  \Omega_{\mathrm{RD}}^2 \Omega_{\mathrm{SR}}} \\
		\nonumber &+\frac{(\eta_1 \tau_1  \Omega_{\mathrm{SR}}+\eta_3 (1-\tau_1 ) \Omega_{\mathrm{RD}})^2}{4 \gamma_P \eta_1 \eta_2 \eta_3 \tau_1  \Omega_{\mathrm{RD}}^2 \Omega_{\mathrm{SR}}}, \\
		\nonumber	\psi_7=&\frac{\eta_1 \Omega_{\mathrm{SR}} (\eta_2-\eta_3 \tau_1 +\eta_3)-\eta_1 (\gamma_P \eta_2 \eta_3 \Omega_{\mathrm{RD}} \psi_1+\Omega_{\mathrm{SR}} \sqrt{\psi_2})}{2 \gamma_P \eta_1 \eta_2 \eta_3 \Omega_{\mathrm{RD}} \Omega_{\mathrm{SR}}}
		\nonumber \\&+\frac{\Omega_{\mathrm{RD}} (-\eta_2-\eta_3) (\eta_2 (\tau_1 -1)+\eta_3 (2 \tau_1 -1))}{2 \gamma_P \eta_1 \eta_2 \eta_3 \Omega_{\mathrm{RD}} \Omega_{\mathrm{SR}}},\\
		\nonumber \psi_8 =&\frac{\eta_2^2 (\tau_1-1)+\eta_3^2 (2 \tau_1-1)+\eta_2 \eta_3 (3 \tau_1+\gamma_P \eta_1 \psi_1-2)}{2 \gamma_P \eta_1 \eta_2 \eta_3 \Omega_{\mathrm{SR}}}.
		\end{align*}
		
	\end{proposition}
	
	\begin{proof}
		The proof is provided in appendix~\ref{app:A}.
	\end{proof}
	
	\subsection{Average Fractional Equivocation -- AFE}
	
	{  By taking the average of the fractional equivocation in~\eqref{eq:Deltau_1br} we can derive the (longterm) average for the factional equivocation. Therefore, as the eavesdropper's decoding error probability for a given fading realization is asymptotically lower bounded by the fractional equivocation, the average factional equivocation, $\bar\Delta$ , actually provides an asymptotic lower bound on the overall decoding error probability at Eve, thus being an error-probability-based secrecy metric, given as in~\cite{art:he}}
	\begin{align}\label{eq:afefor}	
	\bar \Delta =& \mathbb{E}\{\Delta\}	.
	\end{align}
	For the proposed system, an approximation for the average fractional equivocation at high SNR is obtained as in Proposition~\ref{prop:afe}.
	
	\begin{proposition}\label{prop:afe}
		An approximate expression for the average fractional equivocation of a three-node AF relaying network with an untrusted relay is given by
		\begin{align}\label{eq:afep}
		\nonumber	\bar \Delta \approx& 1-\left(1-\dfrac{\ln(2^{2 R_S})}{\ln(2){2 R_S}}\right)F_{\Phi}\left(2^{2 R_S}\right)\\
		&-\dfrac{1}{\ln(2){2 R_S}}\left[  \ln(1) F_{\Phi}\left(1\right) +\int_{1}^{2^{2 R_S}}\dfrac{1}{\phi}F_{\Phi}\left(\phi\right)d \phi\right] ,	
		\end{align}
		where  $F_{\Phi}\left(\cdot\right)$ is the CDF of the random variable $\Phi$ defined as in~\eqref{eq:Deltau_1br}, given by~\eqref{eq:cdfpphi1} in Appendix~\ref{app:A}. \eqref{eq:gsopfinal}, which was obtained as described in Appendix~\ref{app:A}.
	\end{proposition}
	\begin{proof}
		The proof is provided in Appendix~\ref{app:D}.
	\end{proof}

	\subsection{Average Information Leakage Rate -- AILR}
	
	{  For the cases where classical information-theoretic secrecy is not achievable, some information will be leaked to the eavesdropper. Therefore, different secure transmission schemes that lead to the same secrecy outage probability may actually result in very different amounts of information leakage. Then, this metric provides a notion of how fast the information is leaked to the untrusted relay. Thus, for a fixed-rate transmission, it can be defined as in~\cite{art:he}}
	\begin{align}\label{eq:rl}
	R_L=\left(1-\bar \Delta\right) R_S.
	\end{align}
	For the proposed system, an approximation for the average information leakage rate at high SNR is obtained by following Proposition~\ref{prop:afe} and~\eqref{eq:rl}, as follows.
	
	\begin{corollary}\label{prop:ailr}
		An approximate expression for the average information leakage rate of a three-node AF relaying network with an untrusted relay is given by
		\begin{align}\label{eq:ailrfin}
		\nonumber	R_L \approx& R_S\left(1-\dfrac{\ln(2^{2 R_S})}{\ln(2)2^{2 R_S}}\right)F_{\Phi}\left(2^{2 R_S}\right)\\
		&+\dfrac{R_S}{ \ln(2)2^{2 R_S}}\left[  \ln(1) F_{\Phi}\left(1\right) +\int_{1}^{2^{2 R_S}}\dfrac{1}{\phi}F_{\Phi}\left(\phi\right)d \phi\right],	
		\end{align}
		where  $F_{\Phi}\left(\cdot\right)$ is given by~\eqref{eq:cdfpphi1} in Appendix~\ref{app:A}.
	\end{corollary}
	{ 
	\section{ Asymptotic Analysis}
	\label{sec:asymptotic}
	In this section, simpler analytical expressions for the three metrics analyzed above are obtained by considering the system performance at high SNR, which are useful to grasp a better insight into the attained diversity order and the impact of key parameters on the secrecy performance of the investigated system. These asymptotic expressions for GSOP, AFE, and AILR are introduced in the following proposition.
	\begin{proposition}\label{prop:asymptote}
	Asymptotic expressions for the GSOP~\eqref{eq:gsopfinal}, AFE~\eqref{eq:afep}, and AILR~\eqref{eq:ailrfin} of a three-node AF relaying network with an untrusted relay are given by
	\begin{align}
	 \nonumber \mathrm{GSOP}^{\infty}=&\sqrt{\frac{\pi\eta_1 \tau_1  \Omega_{\mathrm{SR}}}{4 \gamma_P \eta_2 \eta_3 \Omega_{\mathrm{RD}}^2}}\\
	\nonumber &+\frac{2 (1\!-\!\eta_1) \eta_3 (\tau_1 \!-\!1) \Omega_{\mathrm{RD}}\!+\!\eta_1 \Omega_{\mathrm{SR}} (\eta_1\!+\!\eta_3 \tau_1\! -\!1)}{2 \gamma_P \eta_1 \eta_2 \eta_3 \Omega_{\mathrm{RD}} \Omega_{\mathrm{SR}}}\\
	 &\propto \left(\dfrac{1}{\gamma_P}\right)^{\tfrac{1}{2}},\label{eq:gsopasymp}\\
\nonumber	 \bar\Delta^{\infty} \!\!= & \frac{1}{2 \gamma_P \eta_1 \eta_2 \eta_3 \Omega_{\mathrm{RD}} \Omega_{\mathrm{SR}} \log (2^{2 R_S} )}\\
\nonumber	 & \times \left[ -2 \sqrt{\pi } \sqrt{\gamma_P} \eta_1^{3/2} \sqrt{\eta_2} \sqrt{\eta_3} \left(\sqrt{2^{2 R_S} }-1\right) \Omega_{\mathrm{SR}}^{3/2} \right.+\!\ln (2^{2 R_S} )\\
\nonumber	 &\times\!\left[2 \eta_3 \Omega_{\mathrm{RD}} (\gamma_P \eta_1 \eta_2 \Omega_{\mathrm{SR}}\!-\!\eta_1\!+\!1)\!-\!(\eta_1\!-\!1) \eta_1 \Omega_{\mathrm{SR}}\right]\\
	 & \left. +\eta_3 (2^{2 R_S} -1) (2 (\eta_1-1) \Omega_{\mathrm{RD}}-\eta_1 \Omega_{\mathrm{SR}})\right],\label{eq:afeasymp}\\
R_L^{\infty}=&\left(1-\bar\Delta^{\infty}\right) R_S.
	\end{align}
	\end{proposition}
	
	\begin{proof}
		The asymptotic expressions are derived by considering the high SNR regime, i.e., with $\gamma_P \rightarrow \infty$. Then, by neglecting the higher order terms of the Maclaurin series expansion for the exponential function in~\cite[Eq. (0.318.2)]{book:gradshteyn}, we have that $e^{-x}\simeq1-x$. After replacing this into the CDFs of expressions~\eqref{eq:T1fin},~\eqref{eq:T2fin},~\eqref{eq:T3fin}, and~\eqref{eq:T4fin} in appendix~A, solving the integrals, and performing some simplifications, the expression in~\eqref{eq:gsopasymp} can be obtained. can be obtained. Finally, by replacing~\eqref{eq:gsopasymp} into~\eqref{eq:afep}, it yields the expression in~\eqref{eq:afeasymp}. 
	\end{proof}
	\begin{remark}
	From~\eqref{eq:gsopasymp}, it can be infer that the system diversity order equals $\sim\tfrac{1}{2}$.
	\end{remark}
	}
	
	\section{ Throughput-Constrained Secrecy Performance Analysis}
	\label{sec:throughput}
	
	In this section, the previous secrecy performance metrics are optimized by taking into account a minimum throughput of confidential transmission constraint, $\Gamma$. { We define the throughput of confidential transmission, $\mathcal{T}$, as the probability of successful transmission times the secrecy rate. Then, an approximate anlytical expression for the throughput of the confidential transmission can be obtained as stated in Proposition~\ref{prop:throughput}}
	
	\begin{proposition}\label{prop:throughput}
		An approximation for the throughput of confidential transmission for a three-node AF relaying network with an untrusted relay is given by
		\begin{align}\label{eq:throughput}
		\mathcal{T} \approx R_S \exp\left[-\frac{\left(2^{2 R_T}-1\right) (\eta_1 \Omega_{\mathrm{SR}}+\Omega_{\mathrm{RD}} (\eta_2+\eta_3))}{\gamma_P \eta_1 \eta_2 \Omega_{\mathrm{RD}} \Omega_{\mathrm{SR}}}\right],
		\end{align}
		where $R_T$ is the codeword transmission rate.
	\end{proposition}
	\begin{proof}
		The proof is provided in Appendix~\ref{app:B}.
	\end{proof}
	\begin{proposition}
		\label{cor:thrfeas}
		The maximum achievable throughput constrained to $R_T\geq R_S>0$, for a given power allocation factor at D, $\eta_3$, is given by
		\begin{align}
		\mathcal{T_{\mathrm{max}}}\!=\! R_{S}^{\mathcal{T}} \exp\!\left[\frac{\left(1\!-\!2^{2 R_{S}^{\mathcal{T}}}\right) (\Omega_{\mathrm{RD}} ( \eta_2^{\mathcal{T}}\!+\! \eta_3)\!+\! \Omega_{\mathrm{SR}} (1\!-\!\eta_2^{\mathcal{T}}\!-\!\eta_3))}{\gamma_P \eta_2^{\mathcal{T}} \Omega_{\mathrm{RD}} \Omega_{\mathrm{SR}} (1\!-\!\eta_2^{\mathcal{T}}\!-\!\eta_3)}\right],
		\end{align}
		where 
		\begin{scriptsize}
		\begin{align}
		R_{S}^{{\mathcal{T}}}&=\dfrac{1}{ \ln 4}W\left(-\frac{\gamma_P \eta_2^{\mathcal{T}} \Omega_{\mathrm{RD}} \Omega_{\mathrm{SR}} (1-\eta_2^{\mathcal{T}}-\eta_3)}{(\eta_2^{\mathcal{T}}+\eta_3) (\Omega_{\mathrm{RD}}-\Omega_{\mathrm{SR}})+\Omega_{\mathrm{SR}}}\right), \label{eq:Rsmaxt}\\
		\eta_2^{\mathcal{T}}&\!=\!\left\lbrace 
		\begin{array}{ll}
		\tfrac{1}{2}-\tfrac{1}{2}\eta_3, & \mathrm{if}\ \Omega_{\mathrm{SR}}\!=\! \Omega_{\mathrm{RD}}  \\ 
		-\sqrt{\frac{\Omega_{\mathrm{RD}} (\eta_3 (\Omega_{\mathrm{RD}}-\Omega_{\mathrm{SR}})+\Omega_{\mathrm{SR}})}{(\Omega_{\mathrm{RD}}-\Omega_{\mathrm{SR}})^2}}\!-\!\eta_3\!-\!\frac{\Omega_{\mathrm{SR}}}{\Omega_{\mathrm{RD}}-\Omega_{\mathrm{SR}}} &  \mathrm{if}\ \Omega_{\mathrm{SR}}\!>\! \Omega_{\mathrm{RD}}\\
		\sqrt{\frac{\Omega_{\mathrm{RD}} (\eta_3 (\Omega_{\mathrm{RD}}-\Omega_{\mathrm{SR}})+\Omega_{\mathrm{SR}})}{(\Omega_{\mathrm{RD}}-\Omega_{\mathrm{SR}})^2}}\!-\!\eta_3\!-\!\frac{\Omega_{\mathrm{SR}}}{\Omega_{\mathrm{RD}}-\Omega_{\mathrm{SR}}} & \mathrm{if}\ \Omega_{\mathrm{SR}}\!<\! \Omega_{\mathrm{RD}}  .
		\end{array}  \right. \label{eq:etamax}
		\end{align}
		\end{scriptsize}
	\end{proposition}
	\begin{proof}
		The proof is provided in Appendix~\ref{app:C}.
	\end{proof}
	From Proposition~\ref{cor:thrfeas}, it can be noticed that the feasible range for the throughput of confidential information is $0$$<$$\Gamma$$<$$\mathcal{T_{\mathrm{max}}}$. { In light of this, the system parameters can be optimized regarding the proposed secrecy metrics given a minimum required throughput of confidential information, $\Gamma$, which is directly related to the system reliability. That is, by considering the following optimization problems, we can determine the optimal system parameters required to attain the best secrecy and reliable performance for a certain throughput.}
	
	\subsection{Minimization of the Generalized Secrecy Outage Probability}
	For this optimization problem, we would like to find the rate parameters and the power allocation factors that minimize the generalized secrecy outage probability given in~\eqref{eq:gsopfinal}. Therefore, the optimization problem can be formulated as 
	\begin{itemize}
		\item OPA1:
	\begin{align}\label{eq:mingsop}
	&\min_{R_S, R_T, \eta_1,\eta_2,\eta_3}  \mathrm{GSOP}\left(\tau_1\right) \nonumber\\
	\mathrm{s.t.} \quad & \mathcal{T}>\Gamma, R_T\geq R_S>0, \eta_1>0,\eta_2>0,\eta_3>0,\nonumber\\  &\eta_1+\eta_2+\eta_3=1.
	\end{align}
\end{itemize}
	\subsection{Maximization of the Average Fractional Equivocation}
	Following the same reasoning, we can find the rate parameters and the power allocation factors that maximizes the average fractional equivocation, $\bar \Delta$, given in~\eqref{eq:afep}. Therefore, the optimization problem can be formulated as 
	\begin{itemize}
		\item OPA2:
	\begin{align}
	\max_{R_S, R_T, \eta_1,\eta_2,\eta_3} & \bar\Delta \nonumber\\
	\mathrm{s.t.} \quad & \mathcal{T}>\Gamma, R_T\geq R_S>0, \eta_1>0,\eta_2>0,\eta_3>0,\nonumber\\  &\eta_1+\eta_2+\eta_3=1.\label{eq:maxafe}
	\end{align}
\end{itemize}
	\subsection{Minimization of the Average Information Leakage Rate}
	Finally, we can find the rate parameters and the power allocation factor that minimizes the average fractional equivocation $R_L$, given in~\eqref{eq:rl}. Therefore, the optimization problem can be formulated as
	\begin{itemize}
		\item OPA3:
	\begin{align}
	\min_{R_S, R_T, \eta_1,\eta_2,\eta_3} & R_L \nonumber\\
	\mathrm{s.t.} \quad & \mathcal{T}>\Gamma, R_T\geq R_S>0, \eta_1>0,\eta_2>0,\eta_3>0,\nonumber\\  &\eta_1+\eta_2+\eta_3=1.\label{eq:minailr}
	\end{align}
\end{itemize}
	
	Due to the complexity of the expressions in~\eqref{eq:gsopfinal} and~\eqref{eq:afep}, the optimal values for  $R_{S}^1$,  $R_{T}^1$, $\eta_{1}^1$, $\eta_{2}^1$, and $\eta_{3}^1$, corresponding to the problem OPA1 in~\eqref{eq:mingsop}; $R_{T}^2$, $\eta_{1}^2$, $\eta_{2}^2$, and $\eta_{3}^2$, corresponding to the problem OPA2 in~\eqref{eq:maxafe}; and $R_{S}^3$,  $R_{T}^3$, $\eta_{1}^3$, $\eta_{2}^3$, and $\eta_{3}^3$, corresponding to the problem OPA3 in~\eqref{eq:minailr}, cannot be obtained in closed form. However, a numerical optimization can be performed through optimization tools provided by specialized software such as Wolfram Mathematica or Matlab (e.g. the function \texttt{FindMinimum} of Wolfram Mathematica). Alternatively, by considering the asymptotic expressions in~\eqref{eq:gsopasymp} and~\eqref{eq:afeasymp}, very close results to those obtained by numerical optimization via software tools can be attained by resorting to iterative algorithms to solve optimization problems of continuous nonlinear functions. For instance, in Algorithm~\ref{alg:pso} we apply the Particle Swarm Optmization (PSO) method (cf.~\cite{art:Kennedy1995} for more details), which is a population-based stochastic optimization algorithm, to solve the formulated optimization problems. This algorithm offers a simpler implementation that can be applied in practical systems.

	\begin{algorithm}
	\caption{PSO algorithm}\label{alg:pso}
	\begin{algorithmic}[1]
		\begin{scriptsize}
			\STATE\textbf{function} \:{\texttt{OF}}({$\mathtt{var},\gamma_P, \alpha, \theta, \Omega_{\mathrm{SR}}, \Omega_{\mathrm{RD}}, \Gamma$}) \COMMENT{$\mathtt{var}$ is an array}
			\STATE $\mathtt{of} \gets$ Eq.~\eqref{eq:gsopasymp} \COMMENT{Cost function of GSOP, (-)AFE, or AILR}
			\STATE $\mathtt{penalty} \gets 1$
			\STATE $\mathtt{const1} \gets$ Eq.~\eqref{eq:throughput}
			\STATE $\mathtt{const2} \gets\mathtt{var}(1)+\mathtt{var}(2)+\mathtt{var}(3)$
			\IF {$\mathtt{const1} > \Gamma $ }
			\STATE $\mathtt{penalty} \gets  penvalue$
			\ENDIF
			\IF {$\mathtt{const2}  != 1$ }
			\STATE $\mathtt{penalty} \gets  \mathtt{penalty} *penvalue$
			\ENDIF
			\IF {$\mathtt{var}(5) < \mathtt{var}(4)$ }
			\STATE $\mathtt{penalty} \gets  \mathtt{penalty} *penvalue$
			\ENDIF
			\STATE   $\mathtt{of} \gets \mathtt{of}+\mathtt{penalty}$ 
			\STATE \textbf{return} $\mathtt{of}$
			\STATE \textbf{end fucntion}
		\end{scriptsize}
		\begin{scriptsize}
			\STATE \textbf{function}\:{\texttt{Pso}}({$\gamma_P, \alpha, \theta, \Omega_{\mathrm{SR}}, \Omega_{\mathrm{RD}}, \Gamma$})
			\STATE $\mathtt{iter} \gets \# iterations$
			\STATE $\mathtt{p} \gets \# particles$  \COMMENT{Set population size}
			\STATE $\mathtt{c1} \gets c\_init$ \COMMENT{Weighting coeff. for personal best pos.}
			\STATE $\mathtt{c2} \gets c\_init$ \COMMENT{Weighting coeff. for global best pos.}
			\STATE $\mathtt{w} \gets wvalue$ \COMMENT{Set inertia weight}
			\STATE $\mathtt{n} \gets \#variables$ 
			\STATE $\mathtt{varmin}$ $\gets$ $[\mathtt{eta1min} \mathtt{eta2min} \mathtt{eta3min}  \mathtt{Rsmin} \mathtt{Rtmin}]$\COMMENT{Set min values} 
			
			\STATE $\mathtt{varmax}\gets [\mathtt{eta1max} \mathtt{eta2max} \mathtt{eta3max} \mathtt{Rsmax} \mathtt{Rtmax}]$  \COMMENT{Set max values}
			\STATE $\mathtt{v} \gets initialv$ \COMMENT{Set initial velocity}
			\STATE $\mathtt{fgbest} \gets 10$
			\FOR {$\mathtt{i} \gets 1$ to $\mathtt{p}$}
			\FOR {$\mathtt{j} \gets 1$ to $\mathtt{n}$}
			\STATE $\mathtt{var}(\mathtt{i},\mathtt{j}) \gets (\mathtt{varmax}(\mathtt{j})-\mathtt{varmin}(\mathtt{j}))*\textbf{rand}()+\mathtt{varmin}(\mathtt{j})$ \COMMENT{Initial position}
			\STATE $\mathtt{aux}(\mathtt{j}) \gets \mathtt{aux}(\mathtt{i},\mathtt{j})$
			\ENDFOR
			\STATE $\mathtt{fitness}(\mathtt{i}) \gets \mathtt{OF}(\mathtt{aux},\gamma_P, \alpha, \theta, \Omega_{\mathrm{SR}}, \Omega_{\mathrm{RD}}, \Gamma)$
			\IF {$\mathtt{fitness}(\mathtt{i}) < \mathtt{fgbest}$ }
			\STATE $\mathtt{fgbest} \gets  \mathtt{fitness}(\mathtt{i})$
			\FOR {$\mathtt{j} \gets 1 \mathrm{to} \mathtt{n}$}
			\STATE $\mathtt{vargbest}(\mathtt{j}) \gets  \mathtt{var}(\mathtt{i},\mathtt{j})$
			\ENDFOR
			\ENDIF
			\ENDFOR
			\STATE $\mathtt{varpbest} \gets  \mathtt{var}$
			\STATE $\mathtt{fpbest} \gets  \mathtt{fitness}$
			\FOR {$\mathtt{i} \gets 1 \mathtt{to} \mathtt{iter}$}
			\STATE $\mathtt{v} \gets \mathtt{w}*\mathtt{v}+\mathtt{c1}* \textbf{{rand}}(\mathtt{p}, \mathtt{n}) * (\mathtt{varpbest}-\mathtt{var})+ \mathtt{c2}*\textbf{rand}(\mathtt{p}, \mathtt{n}) * (\mathtt{vargbest}-\mathtt{var})) $ \COMMENT{Update velocity (element-wise operation)}
			\STATE $\mathtt{var} \gets \mathtt{var}+\mathtt{v}$ \COMMENT{Update position}
			\FOR {$\mathtt{j} \gets 1\mathrm{to}\mathtt{p}$}
			\FOR {$\mathtt{k} \gets 1\mathrm{to}\mathtt{n}$}
			\IF {$\mathtt{var}(\mathtt{j},\mathtt{k}) < \mathtt{varmin}(\mathtt{k})$ }
			\STATE $\mathtt{var}(\mathtt{j},\mathtt{k}) \gets  \mathtt{varmin}(\mathtt{k})$
			\ENDIF
			\IF {$\mathtt{var}(\mathtt{j},\mathtt{k}) > \mathtt{varmax}(\mathtt{k})$ }
			\STATE $\mathtt{var}(\mathtt{j},\mathtt{k}) \gets  \mathtt{varmax}(\mathtt{k})$
			\ENDIF
			\STATE $\mathtt{aux}(\mathtt{k}) \gets \mathtt{aux}(\mathtt{j},\mathtt{k})$
			\ENDFOR
			\STATE $\mathtt{fitness}(\mathtt{j}) \gets \mathtt{OF}(\mathtt{aux},\gamma_P, \alpha, \theta, \Omega_{\mathrm{SR}}, \Omega_{\mathrm{RD}}, \Gamma)$
			\IF {$\mathtt{fitness}(\mathtt{j}) < \mathtt{fpbest}(\mathtt{j})$ }
			\STATE $\mathtt{fpbest}(\mathtt{j}) \gets  \mathtt{fitness}(\mathtt{j})$
			\FOR {$\mathtt{k} \gets 1\mathrm{to}\mathtt{n}$}
			\STATE $\mathtt{varpbest}(\mathtt{j},\mathtt{k}) \gets  \mathtt{var}(\mathtt{j},\mathtt{k})$ \COMMENT{Up. personal best}
			\ENDFOR
			\ENDIF
			\IF {$\mathtt{fitness}(\mathtt{j}) < \mathtt{fgbest}$ }
			\STATE $\mathtt{fgbest} \gets  \mathtt{fitness}(\mathtt{j})$
			\FOR {$\mathtt{k} \gets 1\mathrm{to}\mathtt{n}$}
			\STATE $\mathtt{vargbest}(\mathtt{k}) \gets  \mathtt{var}(\mathtt{j},\mathtt{k})$ \COMMENT{Up. global best}
			\ENDFOR
			\ENDIF
			\ENDFOR
			\STATE $\mathtt{w} \gets  \mathtt{w}*0.7$ \COMMENT{Decrease inertia weight}
			\ENDFOR
			\STATE \textbf{return} $\mathtt{fgbest}, \mathtt{vargbest}$
			\textbf{end fucntion}
		\end{scriptsize}	
	\end{algorithmic}
\end{algorithm}

	\section{Numerical Results and Discussions}
	\label{sec:results}
	In this section, Monte Carlo simulations are carried out to validate our analytical results for some illustrative cases. For this purpose, it is considered a two-dimensional network topology, where S and D are located at the coordinates $(0,0)$ and $(1,0)$ (assuming normalized distances), respectively, while the untrusted relay is located midway between S and~D. Without loss of generality, it is assumed that the average channel gain for all links is determined by the distance between the respective pair of nodes, i.e., $\Omega_{\mathrm{AB}}=d_{\mathrm{AB}}^{-\alpha}$, with  $\mathrm{A}\mathop{\in}\{\mathrm{S}, \mathrm{R}, \mathrm{D}\}$ and $\mathrm{B}\mathop{\in}\{\mathrm{R},\mathrm{D}\}$, where $d_{\mathrm{AB}}$ is the distance between the corresponding nodes, and $\alpha$ is the path loss exponent (for the evaluated cases, we consider $\alpha=4$). { Moreover, for optimizations we have considered Algorithm~\ref{alg:pso} for a population size of $\# particles= 2000$, $\# iterations=100$, $c\_init=0.3$, and $wvalue=0.3$}. 
	
	Fig.~\ref{fig:gsopvssnr} illustrates the GSOP versus the transmit system SNR for different values of $\theta$, considering $R_S$$=$$1$ bps/Hz, $d=d_{\mathrm{SR}}/d_{\mathrm{SD}}=0.5$ (such that $\Omega_\mathrm{SR}$$=$$\Omega_\mathrm{RD}$), and equal power allocation (EPA) among nodes.
	It can be observed that our approximation is very tight from medium-to-high SNR. Also, notice that, for achieving the same performance in terms of GSOP, a gain of $\sim$6 dB in SNR is obtained when comparing the case in which the capability of the relay to decode the confidential message is very low (e.g. $\theta =0.1$) and that of the classical approach (i.e., $\theta=1$). That is, in scenarios where it is possible a relaxation on the secrecy requirement due to a lower capability of R on decoding the confidential information, important power savings can be obtained by considering the GSOP as a design criterion.
	\begin{figure}[t]\centering
		\psfrag{Generalized Secrecy Outage Probability, GSOP}[Bc][Bc][0.7]{Generalized Secrecy Outage Probability, GSOP}
		\psfrag{System SNR (dB)}[Bc][Bc][0.8]{Transmit system SNR, $\gamma_P$ (dB)}
		\psfrag{theta}[Br][Br][0.8]{$\theta=0.1,0.5,1$}	
		\psfrag{Simulationnn}[Bl][Bl][0.7]{Simulation}	
		\psfrag{Analytical}[Bl][Bl][0.7]{Analytical}
		\psfrag{Asymptotic}[Bl][Bl][0.7]{Asymptotic}
		\psfrag{Average Fractional Equivocation, Deltaav}[Bc][Bc][0.7]{\color{Chocolate3}Average Fractional Equivocation, $\bar \Delta$}
		\includegraphics[width=\linewidth]{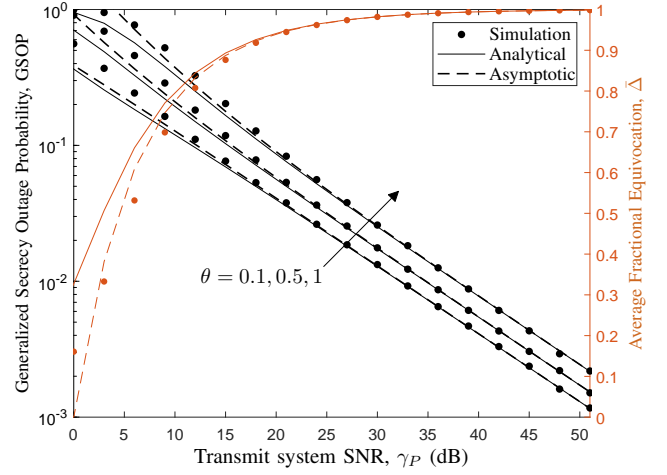}
		\caption{GSOP versus transmit system SNR for different values of $\theta$, considering $R_S=1$ bps/Hz, $\Omega_\mathrm{SR}=\Omega_\mathrm{RD}$, and EPA.}
		\label{fig:gsopvssnr}
	\end{figure}
	
	\begin{figure}[t]\centering
		\psfrag{Generalized Secrecy Outage Probability, GSOP}[Bc][Bc][0.7]{Generalized Secrecy Outage Probability, GSOP}
		\psfrag{Secrecy Rate, RS}[Bc][Bc][0.8]{Secrecy Rate, $R_S$ (bps/Hz)}
		\psfrag{theta}[Br][Br][0.8]{$\theta=1, 0.8, 0.5, 0.2$}	
		\psfrag{Simulationnnnn}[Bl][Bl][0.7]{Simulation}	
		\psfrag{EPA}[Bl][Bl][0.7]{EPA}	
		\psfrag{OPA1}[Bl][Bl][0.7]{OPA1}	
		\psfrag{OPA2}[Bl][Bl][0.7]{OPA2}	
		\psfrag{OPA3}[Bl][Bl][0.7]{OPA3}	
		\psfrag{theta 1}[Bl][Bl][0.8]{$\theta$$=$$1$}
		\psfrag{theta 2 }[Bl][Bl][0.8]{$\theta$$=$$0.5$}
		\psfrag{theta 3 }[Bl][Bl][0.8]{$\theta$$=$$0.1$}		
		\psfrag{OPAREF}[Bl][Bl][0.7]{OPA~\cite{art:Kuhestani2016}}	
		\includegraphics[width=\linewidth]{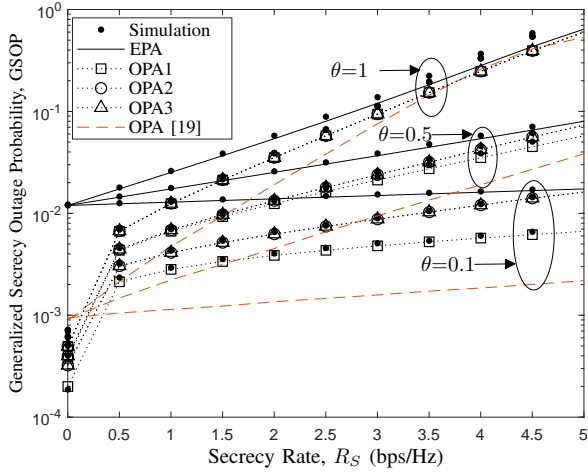}
		\caption{GSOP versus secrecy rate $R_S$ for different values of $\theta$, considering $\gamma_P = 30$ dB and $\Omega_\mathrm{SR}=\Omega_\mathrm{RD}$. A comparison with the OPA scheme in~\cite{art:Kuhestani2016} and the EPA scheme is also shown.}
		\label{fig:gsopvsrs}
	\end{figure}
	
	Fig.~\ref{fig:gsopvsrs} shows the GSOP versus the secrecy rate $R_S$ for different values of $\theta$, considering $\gamma_P$$=$$30$ dB and $d$$=$$0.5$ (such that $\Omega_\mathrm{SR}$$=$$\Omega_\mathrm{RD}$). { This figure also compares different power allocation strategies, namely, EPA, the proposed OPA schemes obtained from the optimization problems in Section~\ref{sec:throughput}, and the OPA presented in~\cite[Eq.~(8)]{art:Kuhestani2016}. First, notice that, even though the OPA scheme of~\cite{art:Kuhestani2016} renders an improved performance regarding our proposed schemes, it noteworthy that the strategy in~\cite{art:Kuhestani2016} is based on instantaneous channel state information (CSI), while our schemes are based on statistical CSI, which makes it more appealing for practical systems, as statistical CSI-based schemes are more robust to imperfect or outdated channel estimation.}
	It can be also observed that, even though an increase on $R_S$ results in a worse secrecy performance for all levels of decodability requirement at R, an increase in $R_S$ results in a more significant loss in secrecy outage performance for the classical approach ($\theta$$=$$1$), thus the secrecy outage probability rapidly achieves 1. However, by relaxing the decodability requirement at R (e.g., for $\theta$$=$$0.5$ and $0.1$), the loss in performance becomes much less pronounced, then higher secrecy rates can be set while attaining the same level of secrecy outage probability. Regarding the power allocation strategies, it is observed that an optimal power allocation scheme greatly improves the secrecy performance leading to ensure secrecy even for higher secrecy rates, and this improvement is more significant as the value of $\theta$ decreases. Additionally, it is observed that as $R_S$ increases, the performance with OPA deteriorates and converges to that with EPA, specially for $\theta=1$, thus, for certain $R_S$, employing an OPA scheme is not advantageous than using equal power allocation factors for all nodes. Also, by comparing the three aforementioned OPA strategies, it is observed that, for this particular case, the three strategies achieve the same performances, except for the case $\theta=0.1$. That is, for scenarios where the decodability at R is low, OPA2 and OPA3 renders the secrecy performance worse than OPA1. All in all, a considerable loss in performance will be attained if the system were designed with  a power allocation that optimize the classical secrecy outage for scenarios that requires a low decodability at R.
	
	A similar analysis is shown in Fig.~\ref{fig:deltavsrs}, where the secrecy performance is measured by the average fractional equivocation (characterized as an asymptotic lower bound of the eavesdropper's decoding error probability) and the average information leakage rate versus the secrecy rate $R_S$, for $\gamma_P$$=$$30$ dB, $R_S$$=$$1$ bps/Hz, and $d$ (i.e., $\Omega_\mathrm{SR}$$=$$\Omega_\mathrm{RD}$). It is observed that higher values of average fractional equivocation can be attained for the proposed system as $R_S$ decreases; on the contrary, the average information leakage rate increases sharply as $R_S$ increases. The same remarks of Fig.~\ref{fig:gsopvsrs} are corroborated regarding the power allocation schemes, that is, OPA1 for $\theta$$=$$1$, OPA2, and OPA3 all attain the same performance, although the improvement with respect to EPA is not much significant. However, when considering OPA1 for $\theta$$=$$0.1$, the performance is compromised.
	\begin{figure}[t]\centering
		\psfrag{Average Fractional Equivocation, Deltaav}[Bc][Bc][0.7]{Average Fractional Equivocation, $\bar \Delta$}
		\psfrag{Average Information Leakage Rate}[Bc][Bc][0.7]{\color{Chocolate3}Average Information Leakage Rate, $R_L$}
		\psfrag{Secrecy Rate, RS}[Bc][Bc][0.8]{Secrecy Rate, $R_S$ (bps/Hz)}
		\psfrag{theta}[Br][Br][0.8]{$\theta=1, 0.8, 0.5, 0.2$}	
		\psfrag{Simulatiionnn}[Bl][Bl][0.7]{Simulation}	
		\psfrag{EPA}[Bl][Bl][0.7]{EPA}	
		\psfrag{OPA1}[Bl][Bl][0.7]{OPA1}	
		\psfrag{OPA2}[Bl][Bl][0.7]{OPA2}	
		\psfrag{OPA3}[Bl][Bl][0.7]{OPA3}	
		\psfrag{theta 1}[Bl][Bl][0.8]{$\theta$$=$$1$}
		\psfrag{theta 3}[Bl][Bl][0.8]{$\theta$$=$$0.1$}		
		\includegraphics[width=\linewidth]{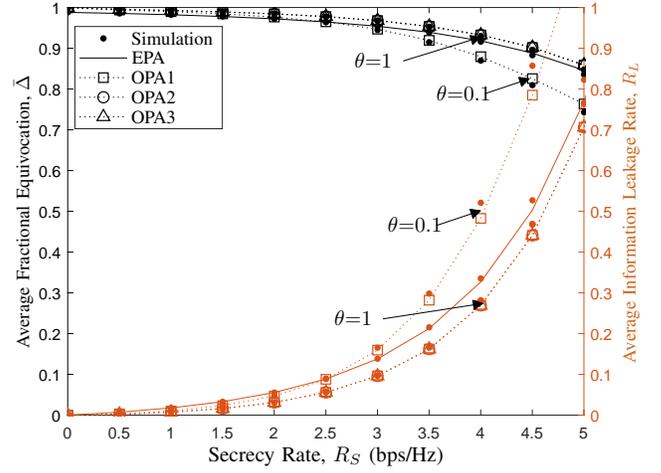}
		\caption{Average fractional equivocation, $\bar \Delta$, and average information leakage rate, $R_L$, versus secrecy rate $R_S$, for $\gamma_P = 30$ dB.}
		\label{fig:deltavsrs}
	\end{figure}
	
	\begin{figure}[t]\centering
		\psfrag{Generalized Secrecy Outage Probability, GSOP}[Bc][Bc][0.7]{Generalized Secrecy Outage Probability, GSOP}
		\psfrag{Distancia}[Bc][Bc][0.8]{Normalized distance, $d $}
		\psfrag{theta}[Br][Br][0.8]{$\theta=1, 0.8, 0.5, 0.2$}	
		\psfrag{Simulationnn}[Bl][Bl][0.7]{Simulation}	
		\psfrag{Analytical}[Bl][Bl][0.7]{EPA}	
		\psfrag{OPA1}[Bl][Bl][0.7]{OPA1}	
		\psfrag{OPA2}[Bl][Bl][0.7]{OPA2}	
		\psfrag{OPA3}[Bl][Bl][0.7]{OPA3}	
		\psfrag{theta 1}[Bl][Bl][0.8]{$\theta$$=$$1,0.5,0.1$}
		\psfrag{theta 2}[Bl][Bl][0.8]{$\theta$$=$$0.1,0.5,1$}
		\psfrag{t3}[Bl][Bl][0.8]{$\theta$$=$$0.1$}		
		\includegraphics[width=\linewidth]{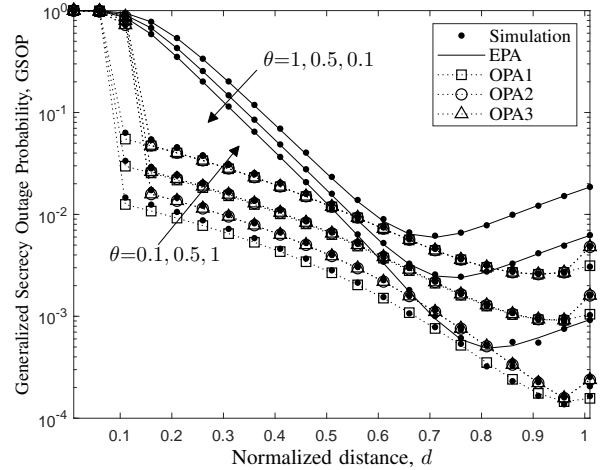}
		\caption{GSOP versus normalized distance $d$, for different values of $\theta$, $R_S=1$ bps/Hz, and $\gamma_P = 30$ dB.}
		\label{fig:gsopvsdist}
	\end{figure}
	
	\begin{figure}[t]\centering
		\psfrag{Average Fractional Equivocation, Deltaav}[Bc][Bc][0.7]{Average Fractional Equivocation, $\bar \Delta$}
		\psfrag{Average Information Leakage Rate}[Bc][Bc][0.7]{\color{Chocolate3}Average Information Leakage Rate, $R_L$}
		\psfrag{Distance}[Bc][Bc][0.8]{Normalized distance, $d  $}
		\psfrag{Secrecy Rate, RS}[Bc][Bc][0.8]{Secrecy Rate, $R_S$ (bps/Hz)}
		\psfrag{theta}[Br][Br][0.8]{$\theta=1, 0.8, 0.5, 0.2$}	
		\psfrag{Simulatiion}[Bl][Bl][0.6]{Simulation}	
		\psfrag{EPA}[Bl][Bl][0.7]{EPA}	
		\psfrag{OPA1}[Bl][Bl][0.7]{OPA1}	
		\psfrag{OPA2}[Bl][Bl][0.7]{OPA2}	
		\psfrag{OPA3}[Bl][Bl][0.7]{OPA3}		
		\psfrag{theta 1}[Bl][Bl][0.8]{$\theta$$=$$1$}
		\psfrag{theta 3}[Bl][Bl][0.8]{$\theta$$=$$0.1$}		
		\includegraphics[width=\linewidth]{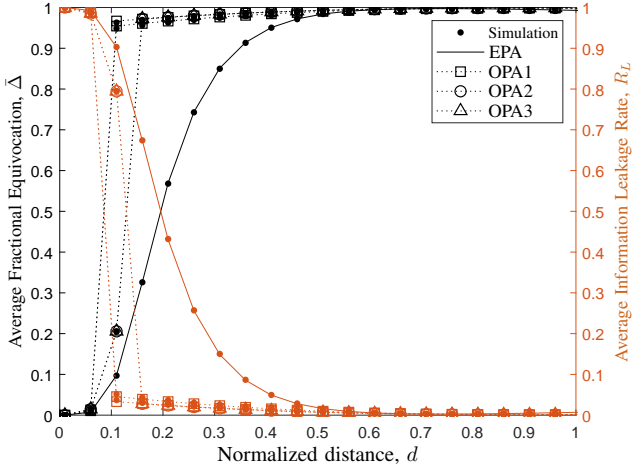}
		\caption{Average fractional equivocation, $\bar \Delta$, and average information leakage rate, $R_L$, versus normalized distance $d$ for $R_S$$=$$1$ bps/Hz and $\gamma_P = 30$ dB.}
	\label{fig:ailvsdist}
\end{figure}

In Figs.~\ref{fig:gsopvsdist} and~\ref{fig:ailvsdist}, the three considered performance metrics are illustrated versus the normalized distance between R and S, $d$, for $\gamma_P$$=$$30$ dB and $R_S$$=$$1$ bps/Hz. For those figures, EPA is compared with the proposed OPA strategies.
From Fig.~\ref{fig:gsopvsdist}, we can notice that there is an optimal relay position for which the secrecy outage probability results in a minimum, which is closer to D as the decoding capability of R decreases. In fact, in the vicinity of D (i.e., for $\Omega_\mathrm{SR}$$<$$\Omega_\mathrm{RD}$), it can be observed a significant loss in performance when considering the classical secrecy outage, while in the vicinity of S (i.e., for $\Omega_\mathrm{SR}$$>$$\Omega_\mathrm{RD}$), the level of decoding capability at R does not significantly impact the secrecy performance, since, at those positions, R is in advantageous channel conditions for eavesdropping (strong wiretap channel and weak jamming channel). Regarding the power allocation strategies, it can be noticed that OPA1, OPA2, and OPA3 all show the same performance except for the case $\theta=0.1$, as previously observed. Even though the difference between OPA1 and the others strategies is minimal for $\theta=0.1$, OPA2 and OPA3 show a performance worse than EPA at the relay positions around the minimum value of GSOP with EPA. Also, it is observed that a significant improvement is achieved by considering OPA instead of EPA, specially for the positions of R closer to S. At the position of the best EPA performance, it can be observed no difference between OPA and EPA, while OPA boosts the secrecy performance at the vicinity of D much more than EPA, thus the best performance is achieved in positions of R very close to D for all values of $\theta$. From Fig.~\ref{fig:ailvsdist}, we can notice that a remarkable difference between EPA and OPA is observed for positions of R closer to S (i.e., for $\Omega_\mathrm{SR}$$>$$\Omega_\mathrm{RD}$), where the channel conditions are advantageous for eavesdropping, while for the positions of R closer to D (i.e., for $\Omega_\mathrm{SR}$$<$$\Omega_\mathrm{RD}$), the secrecy performance, according to these two metrics, is highly favorable using either OPA or EPA. In other words, the secrecy performance in terms of $\bar \delta$ and $R_L$ always improves as R approaches D. Moreover, by considering OPA, independently of the metric being optimized (OPA1, OPA2, or OPA3), the secrecy performance is greatly improved at any position of R. From the above figures, it can be concluded that the secrecy performance from the GSOP perspective or from the other two metrics perspective leads to different insights regarding the optimal position of the relay and the power allocation strategy.
\begin{figure}[t]\centering
	\psfrag{Generalized Secrecy Outage Probability, GSOP}[Bc][Bc][0.7]{Generalized Secrecy Outage Probability, GSOP}
	\psfrag{Minimum required throughput, Gamma}[Bc][Bc][0.8]{Minimum Required Throughput, $\Gamma$}
	\psfrag{theta}[Br][Br][0.8]{$\theta=1, 0.8, 0.5, 0.2$}	
	\psfrag{Simulationnn}[Bl][Bl][0.7]{Simulation}	
	\psfrag{EPA}[Bl][Bl][0.7]{EPA}	
	\psfrag{OPA11}[Bl][Bl][0.7]{OPA1}	
	\psfrag{OPA2}[Bl][Bl][0.7]{OPA2}	
	\psfrag{OPA3}[Bl][Bl][0.7]{OPA3}
	\psfrag{theta 1}[Bl][Bl][0.8]{$\theta$$=$$1$}
	\psfrag{theta 2}[Bl][Bl][0.8]{$\theta$$=$$0.5$}
	\psfrag{theta 3}[Bl][Bl][0.8]{$\theta$$=$$0.1$}		
	\includegraphics[width=\linewidth]{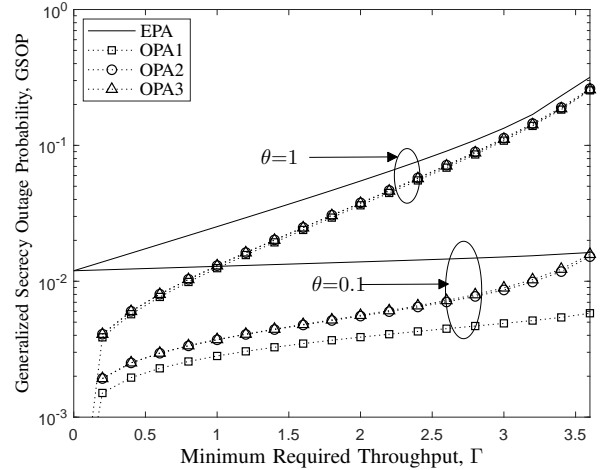}
	\caption{Generalized secrecy outage probability, GSOP, versus minimum required throughput $\Gamma$, for $\theta= 0.1$ and~$1$, and $\gamma_P = 30$ dB.}
	\label{fig:gsopvsgamma}
\end{figure}

\begin{figure}[t]\centering
	\psfrag{Average Fractional Equivocation, Deltaav}[Bc][Bc][0.7]{Average Fractional Equivocation, $\bar \Delta$}
	\psfrag{Minimum required throughput, Gamma}[Bc][Bc][0.8]{Minimum Required Throughput, $\Gamma$}
	\psfrag{Average Information Leakage}[Bc][Bc][0.7]{\color{Chocolate3}Average Information Leakage Rate, $R_L$}
	\psfrag{Secrecy Rate, RS}[Bc][Bc][0.8]{Secrecy Rate, $R_S$ (bps/Hz)}
	\psfrag{theta}[Br][Br][0.8]{$\theta=1, 0.8, 0.5, 0.2$}	
	\psfrag{Simulationnn}[Bl][Bl][0.7]{Simulation}	
	\psfrag{EPA}[Bl][Bl][0.7]{EPA}	
	\psfrag{OPA11}[Bl][Bl][0.7]{OPA1}	
	\psfrag{OPA2}[Bl][Bl][0.7]{OPA2}	
	\psfrag{OPA3}[Bl][Bl][0.7]{OPA3}
	\psfrag{theta 1}[Bl][Bl][0.8]{$\theta$$=$$1$}
	\psfrag{theta 3}[Bl][Bl][0.8]{$\theta$$=$$0.1$}		
	\includegraphics[width=\linewidth]{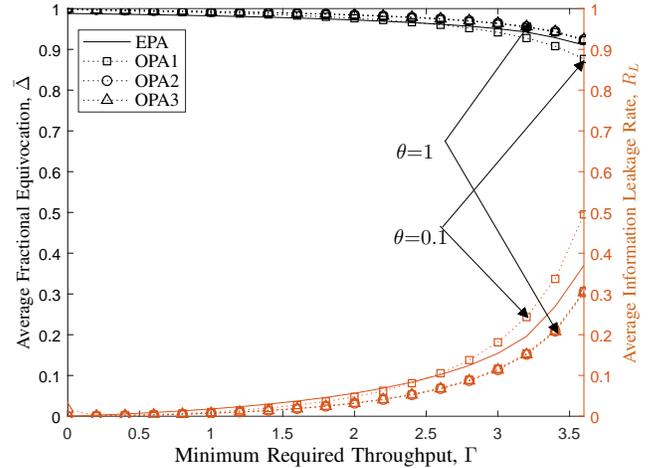}
	\caption{Average fractional equivocation, $\bar \Delta$, and average information leakage rate, $R_L$, versus minimum required throughput $\Gamma$, for $\gamma_P = 30$ dB.}
	\label{fig:afeailrvsgamma}
\end{figure}

Figs.~\ref{fig:gsopvsgamma} and~\ref{fig:afeailrvsgamma} illustrate the three considered metrics versus the minimum required throughput of confidential information, $\Gamma$, for $\gamma_P$$=$$30$ dB. Both figures compare the optimal parameters obtained numerically by solving the optimization problems in~\eqref{eq:mingsop},~\eqref{eq:maxafe}, and~\eqref{eq:minailr}. 


We can observe from~Fig.~\ref{fig:gsopvsgamma} that, by considering GSOP, the secrecy performance deteriorates in a higher rate for the classical approach as the constraint on the throughput of confidential information increases. Also, performing OPA improves the secrecy performance in a more significant manner when the requirement on the decoding capability of R is low (e.g. $\theta{=}0.1$). Besides, as the throughput constraint increases, there is no difference between EPA and OPA for the classical approach, but for lower values of $\theta$, a significant gain on performance by considering OPA1 is observed. On the other hand, regarding the other two metrics, we can observe from Fig.~\ref{fig:afeailrvsgamma} that higher values of $\bar \Delta$ can be achieved until the maximum achievable value of $\mathcal{T}$ is attained ($\sim$ 3.8), which is independent of the adopted power allocation strategy. However, regarding the average information leakage rate, there is a rapidly increasing as $\Gamma$ approaches the maximum achievable value for $\mathcal{T}$, and some differences can be evidenced between the considered power allocation strategies, where a better performance of OPA1 over OPA2, OPA3 is observed for $\theta=0.1$. Therefore, it is evidenced that the allocation of parameters to improve the secrecy performance is not the same regarding the metric that is considered, that is the classical approach, GSOP, AFE, or AILR. 

\begin{figure}[t]\centering
	\psfrag{Throughput, T}[Bc][Bc][0.7]{Throughput, $\mathcal{T}$}
	\psfrag{Secrecy Rate, RS}[Bc][Bc][0.8]{Secrecy Rate, $R_S$ (bps/Hz)}
	\psfrag{RSmax}[Bl][Bl][0.8]{$R_S^{\mathcal{T}}, \eqref{eq:Rsmaxt}$}	
	\psfrag{Simulationnnn}[Bl][Bl][0.7]{Simulation}	
	\psfrag{EPA}[Bl][Bl][0.7]{EPA}	
	\psfrag{eta302}[Bl][Bl][0.6]{$\eta_2^\mathcal{T}, \eta_3=0.2$}	
	\psfrag{eta308}[Bl][Bl][0.6]{$\eta_2^\mathcal{T}, \eta_3=0.8$}	
	\psfrag{pos0.2}[Bl][Bl][0.6]{$\Omega_\mathrm{SR}$$>$$\Omega_\mathrm{RD}$}
	\psfrag{pos0.5}[Bl][Bl][0.6]{$\Omega_\mathrm{SR}$$=$$\Omega_\mathrm{RD}$}	
	\psfrag{pos0.8}[Bl][Bl][0.6]{$\Omega_\mathrm{SR}$$<$$\Omega_\mathrm{RD}$}	
	\psfrag{SOPA3}[Bl][Bl][0.7]{OPA3-$\mathcal{T}$}	
	\psfrag{theta 1}[Bl][Bl][0.8]{$\theta$$=$$1$}
	\psfrag{theta 2}[Bl][Bl][0.8]{$\theta$$=$$0.5$}
	\psfrag{theta 3}[Bl][Bl][0.8]{$\theta$$=$$0.1$}		
	\includegraphics[width=\linewidth]{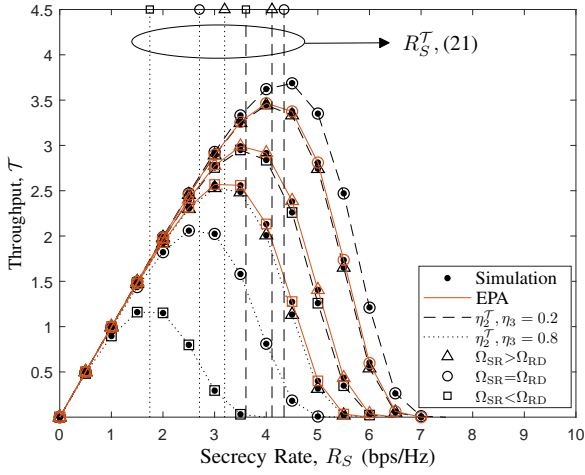}
	\caption{Throughput of confidential transmission versus secrecy rate $R_S$ for the scenarios with EPA and with power allocation according with the optimal value $\eta_2^\mathcal{T}$ from \eqref{eq:etamax}, for different scenarios and $\gamma_P=30$ dB.}
	\label{fig:throughput}
\end{figure}

Finally, Fig.~\ref{fig:throughput} illustrates the throughput of confidential transmission versus the secrecy rate, $R_S$, by comparing EPA and the optimal value for $\eta_2^\mathcal{T}$ from \eqref{eq:etamax}, for different scenarios and $\gamma_P=30$ dB. Notice that, for all cases, it is confirmed that, for the corresponding values of $R_S^\mathcal{T}$ and $\eta_2^\mathcal{T}$, the maximum throughput is attained. Moreover, it can be observed that the best scenario in terms of secrecy throughput is the one where $\Omega_\mathrm{SR}$$=$$\Omega_\mathrm{RD}$ (relay in the midway), and it is allocated less power for jamming. Indeed, for the three scenarios, $\Omega_\mathrm{SR}$$>$$\Omega_\mathrm{RD}$, $\Omega_\mathrm{SR}$$=$$\Omega_\mathrm{RD}$, and $\Omega_\mathrm{SR}$$<$$\Omega_\mathrm{RD}$, the best throughput is achieved when allocating less power to jamming, and EPA performs better than the strategy that considers a higher power for jamming while considering the optimal values for $R_S$ and $\eta_2$. The worst throughput is obtained when $\Omega_\mathrm{SR}$$<$$\Omega_\mathrm{RD}$, and a higher power is intended for jamming.

\section{Conclusions and Future Works}
This paper investigated the secrecy performance on the partial secrecy regime for a three-node amplify-and-forward relaying network with an untrusted relay, in which destination-based jamming is used to enable secure communication. To this end, analytical expressions for three recently proposed secrecy metrics, which are based on the concept of fractional equivocation, were obtained and verified via Monte Carlo simulations. Specifically, a closed-form approximation for the generalized secrecy outage probability was derived, while integral-form approximate expressions were obtained for the average fractional equivocation and the average information leakage rate. Those expressions proved to be very tight to the simulation results at medium-to-high SNR. Numerical results showed that optimal power allocation schemes greatly improve the system secrecy performance, thus reinforcing the advantages of cooperative communications even when the relay is untrusted. However, according to the analyzed metrics, different power allocation strategies and different rate parameters result in different levels of secrecy system performance. Therefore, the optimal power allocation and secrecy rate depend on the secrecy requirements of the system and the capability of the relay in decoding the confidential information. Thus, a significant loss in performance can be obtained when using optimal power allocation according to the classical secrecy outage probability definition, by considering a system with lower requirements on the decoding capability at the relay or lower information leakage rate values. Indeed, the classical approach is highly restrictive in terms of secrecy rate and throughput. Then, for higher values of these two parameters, there is no gain in performance by using optimal power allocation or equal power allocation.

\begin{figure}[t]\centering
		\psfrag{Generalized Secrecy Outage Probability, GSOP}[Bc][Bc][0.7]{Generalized Secrecy Outage Probability, GSOP}
		\psfrag{Secrecy Rate, RS}[Bc][Bc][0.8]{Secrecy Rate, $R_S$ (bps/Hz)}
		\psfrag{theta}[Br][Br][0.8]{$\theta=1, 0.8, 0.5, 0.2$}	
		\psfrag{N=1111111}[Bl][Bl][0.6]{$N=1$}	
		\psfrag{TVT}[Bl][Bl][0.6]{$N=3, [17][21]$}	
		\psfrag{Letterssssssss}[Bl][Bl][0.6]{$N=3, [30]$}	
		\psfrag{theta1}[Bl][Bl][0.8]{$\theta$$=$$1$}
		\psfrag{theta2}[Bl][Bl][0.8]{$\theta$$=$$0.5$}
		\psfrag{theta3}[Bl][Bl][0.8]{$\theta$$=$$0.1$}		
		\psfrag{OPAREF}[Bl][Bl][0.6]{OPA3}	
		\includegraphics[width=\linewidth]{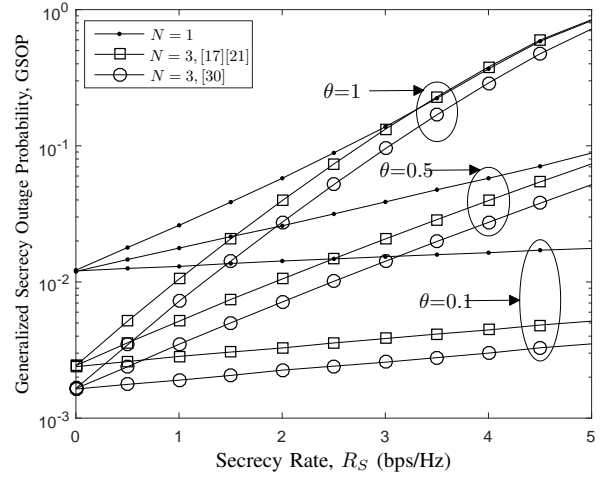}
		\caption{GSOP versus secrecy rate $R_S$ for different values of $\theta$ and number of relays, considering $\gamma_P = 30$ dB and $\Omega_\mathrm{SR}=\Omega_\mathrm{RD}$.}
		\label{fig:gsopvsrs3}
	\end{figure}
	
	\begin{figure}[t]\centering
		\psfrag{Average Fractional Equivocation, Deltaav}[Bc][Bc][0.7]{Average Fractional Equivocation, $\bar \Delta$}
		\psfrag{Average Information Leakage Rate}[Bc][Bc][0.7]{\color{Chocolate3}Average Information Leakage Rate, $R_L$}
		\psfrag{Secrecy Rate, RS}[Bc][Bc][0.8]{Secrecy Rate, $R_S$ (bps/Hz)}
		\psfrag{theta}[Br][Br][0.8]{$\theta=1, 0.8, 0.5, 0.2$}	
		\psfrag{Simulationnnnn}[Bl][Bl][0.6]{$N=1$}	
		\psfrag{REf1}[Bl][Bl][0.6]{$N=3, [17][21]$}	
		\psfrag{Ref2}[Bl][Bl][0.6]{$N=3, [30]$}		
		\includegraphics[width=\linewidth]{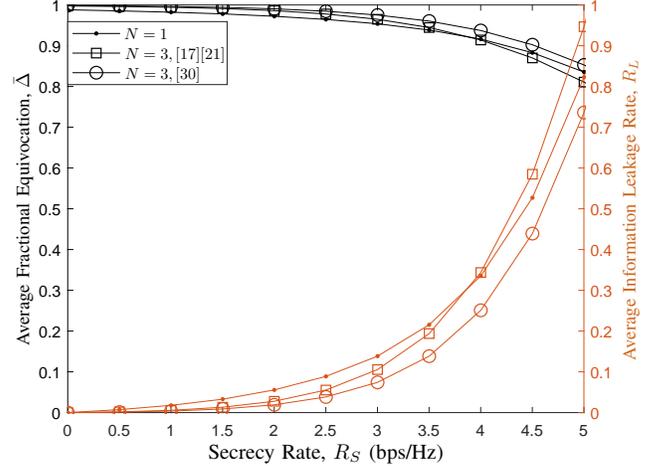}
		\caption{Average fractional equivocation, $\bar \Delta$, and average information leakage rate, $R_L$, versus secrecy rate $R_S$, for $\gamma_P = 30$ dB and different number of relays.}
		\label{fig:deltavsrs3}
	\end{figure}
	
Finally, we point out as future directions the analysis of the multi-relay scenario, where one out of $N$ untrusted relays is selected to cooperate with the legitimate transmission between Alice and Bob. In the case of non-colluding relays, this scenario is similar to that of multiple eavesdroppers, where the eavesdropped information is given by the strongest eavesdropping link. In that case, different scenarios can be tackled. For instance, in~\cite{art:Sun2012} and~\cite{art:osorio2018}, it is assumed that the selected relay uses directional antennas in order to focus the transmission beam toward the destination, during the second transmission phase. Then, the signal received at the non-selected relays from the selected one is weak enough that can be neglected. This way, an information leakage occurs only in the transmission phase from the source. On the other hand, the work in~\cite{art:Sun2015} does not assume directional antennas, thus the non-selected relays can overhear the transmission from the selected relay and try to eavesdrop on it. Thus, the non-selected relays might combine the signals coming from Alice and the selected relay. To counteract this problem, the authors in~\cite{art:Sun2015} proposed a source-based jamming during the second transmission phase in order to enhance the secrecy performance. In Figs.~\ref{fig:gsopvsrs3} and~~\ref{fig:deltavsrs3} are shown some simulations for those protocols, where EPA is considered in all cases, and maximal-ratio combining is considered for the protocol in~\cite{art:Sun2015}. For the multi-relay case, it is considered that the selected relay is the one that maximizes the legitimate link. Note that the jamming and power allocation strategies would play an important role into the gains that can be obtained from the use of multiple relays. Moreover, higher values of $R_S$ are more restricted in the classical secrecy outage probability definition, thus impacting on the amount of information leakage. For instance, due to the jamming signals transmitted during the two phases, it is observed that the average information leakage rate, $R_L$, becomes lower in~\cite{art:Sun2015}. Thus, investigating novel relay selection strategies, protocols and optimal power allocations for the multi-relay scenario would render important insights on the performance of untrusted relay networks in the partial secrecy regime, thereby being a strong motivation for future works.

\vspace{-1cm}

\appendices

\section{Proof of Proposition \ref{prop:gsop}}
\label{app:A}
By replacing~\eqref{eq:Deltau_1br} in~\eqref{eq:gsop}, the GSOP can be expressed as
\begin{align}\label{eq:gsop_phi}
\nonumber \mathrm{GSOP}=& \Pr\left(\Phi \leq 1\right) +\Pr \left(\Phi <2^{2 R_S\theta} \mid  \Phi < 2^{2 R_S} < \Phi 2^{2 R_S} \right) \\
\nonumber & \times \Pr\left(\Phi < 2^{2 R_S} < \Phi 2^{2 R_S}\right)\\
\nonumber \buildrel (a) \over=& F_{\Phi}\left(1\right)+ F_{\Phi}\left(2^{2 R_S\theta}\right) - F_{\Phi}\left(1\right)\\
=&  F_{\Phi}\left(2^{2 R_S\theta}\right),
\end{align}
where $(a)$ was obtained by noticing that $2^{2 R_S\theta} \leq 2^{2 R_S}$. Thus, an expression for the GSOP can be obtained by deriving the CDF of the random variable $\Phi$. However, the exact solution for this problem is mathematically intricate, thus an approximation for this CDF can be obtained by resorting to approximate the harmonic mean by its well known upper bound, $\tfrac{A B}{A + B}< \min\{A,B\}$, as in~\cite{art:osorio}, which renders a tight approximation from medium-to-high SNR. Then,~\eqref{eq:snrD2} can be approximated as 
\begin{align}\label{eq:snrDapprox}
\nonumber \Gamma_{\mathrm{D}} =&  \dfrac{\gamma_{\mathrm{R}}}{\gamma_{\mathrm{R}}+\gamma_{\mathrm{D}}} \dfrac{\gamma_{\mathrm{S}} g_{\mathrm{SR}} \left(\gamma_{\mathrm{R}}+\gamma_{\mathrm{D}}\right) g_{\mathrm{RD}}}{\gamma_{\mathrm{S}} g_{\mathrm{SR}} + \left(\gamma_{\mathrm{R}}+\gamma_{\mathrm{D}}\right) g_{\mathrm{RD}}+1}\\
<&\dfrac{\gamma_{\mathrm{R}}}{\gamma_{\mathrm{R}}+\gamma_{\mathrm{D}}} \min\{\gamma_{\mathrm{S}} g_{\mathrm{SR}},\left(\gamma_{\mathrm{R}}+\gamma_{\mathrm{D}}\right) g_{\mathrm{RD}}\}.
\end{align}

Therefore, a lower bound for the CDF of $\Phi$ can be formulated as
\begin{align}\label{eq:cdfpphi1}
\nonumber	&F_{\Phi}^{\mathrm{LB}}\!\!\left(\phi\right)\!=\! \Pr\left(\dfrac{1+\dfrac{\gamma_{\mathrm{R}}}{\gamma_{\mathrm{R}}+\gamma_{\mathrm{D}}} \min\{\gamma_{\mathrm{S}} g_{\mathrm{SR}},\left(\gamma_{\mathrm{R}}+\gamma_{\mathrm{D}}\right) g_{\mathrm{RD}}\}}{1+ \dfrac{\gamma_{\mathrm{S}} g_{\mathrm{SR}}}{\gamma_{\mathrm{D}}g_{\mathrm{RD}}+1}}<\phi\right)\\
\nonumber &=\!\underbrace{\Pr \!\left(\!\dfrac{\gamma_{\mathrm{S}}\gamma_{\mathrm{R}} g_{\mathrm{SR}}}{\gamma_{\mathrm{R}}\!\!+\!\!\gamma_{\mathrm{D}}}\!\!<\! \phi\!+\!\! \dfrac{\phi \gamma_{\mathrm{S}} g_{\mathrm{SR}}}{\gamma_{\mathrm{D}}g_{\mathrm{RD}}\!+\!\!1}\!-\!1 \Big|  g_{\mathrm{RD}}\!>\!\!\dfrac{\gamma_{\mathrm{S}}g_{\mathrm{SR}}}{\gamma_{\mathrm{R}}\!\!+\!\!\gamma_{\mathrm{D}}} \right)\!\Pr\!\left(\!g_{\mathrm{RD}}\!\!>\!\!\dfrac{\gamma_{\mathrm{S}}g_{\mathrm{SR}}}{\gamma_{\mathrm{R}}\!\!+\!\!\gamma_{\mathrm{D}}} \right) }_{J_1}\\
&+\!\underbrace{\Pr\! \left(\gamma_{\mathrm{R}} g_{\mathrm{RD}}\!\!<\! \phi\!+\! \dfrac{\phi \gamma_{\mathrm{S}} g_{\mathrm{SR}}}{\gamma_{\mathrm{D}}g_{\mathrm{RD}}\!\!+\!\!1}\!\!-\!\!1 \Big|  g_{\mathrm{RD}}\!<\!\!\dfrac{\gamma_{\mathrm{S}}g_{\mathrm{SR}}}{\gamma_{\mathrm{R}}\!\!+\!\gamma_{\mathrm{D}}}  \right)\!\Pr\!\left(g_{\mathrm{RD}}\!<\!\!\dfrac{\gamma_{\mathrm{S}}g_{\mathrm{SR}}}{\gamma_{\mathrm{R}}\!\!+\!\!\gamma_{\mathrm{D}}} \right)}_{J_2},
\end{align}
where the terms $J_1$ and $J_2$ can be reorganized and further derived as showing next. First, by solving $J_1$ in terms of $g_\mathrm{SR}$ and $g_\mathrm{RD}$, this can be splitted into two probability terms, regarding the valid regions for $g_\mathrm{SR}$, and reexpressed as 
\begin{align}\label{eq:j1}
\nonumber &J_1= \Pr\! \bigg(\! g_{\mathrm{RD}}\!>\!\dfrac{\left(\phi\!-\!1\right)\left(\gamma_{\mathrm{R}}\!+\!\gamma_{\mathrm{D}}\right)\!+\! g_{\mathrm{SR}} \gamma_{\mathrm{S}} \phi \left(\gamma_{\mathrm{R}}\!+\!\gamma_{\mathrm{D}}\right)\!-\! g_{\mathrm{SR}} \gamma_{\mathrm{R}} \gamma_{\mathrm{S}} }{g_{\mathrm{SR}} \gamma_{\mathrm{S}} \gamma_{\mathrm{R}}  \gamma_{\mathrm{D}}\! -\! \left(\phi\! -\!1\right)  \gamma_{\mathrm{D}} \left( \gamma_{\mathrm{R}}\!+\!\gamma_{\mathrm{D}}\right)},\\ 
\nonumber & g_{\mathrm{SR}}\! <\! \dfrac{\left(\phi \!-\!1\right)  \left(\gamma_{\mathrm{R}}\! +\! \gamma_{\mathrm{D}}\right) }{\gamma_{\mathrm{S}} \gamma_{\mathrm{R}}} ,
g_{\mathrm{RD}}\!>\!\dfrac{\gamma_{\mathrm{S}}g_{\mathrm{SR}}}{\gamma_{\mathrm{R}}\!+\!\gamma_{\mathrm{D}}}\!  \bigg)\\
\nonumber&+ \Pr\bigg(g_{\mathrm{RD}}\!<\!\dfrac{\left(\phi\!-\!1\right)\left(\gamma_{\mathrm{R}}\!+\!\gamma_{\mathrm{D}}\right)\!+\! g_{\mathrm{SR}} \gamma_{\mathrm{S}} \phi \left(\gamma_{\mathrm{R}}\!+\!\gamma_{\mathrm{D}}\right)\!-\! g_{\mathrm{SR}} \gamma_{\mathrm{R}} \gamma_{\mathrm{S}}   }{g_{\mathrm{SR}} \gamma_{\mathrm{S}} \gamma_{\mathrm{R}}  \gamma_{\mathrm{D}} - \left(\phi -1\right)  \gamma_{\mathrm{D}} \left( \gamma_{\mathrm{R}}+  \gamma_{\mathrm{D}}\right)},\\
\nonumber &g_{\mathrm{RD}}\!\geq\!\!\dfrac{\gamma_{\mathrm{S}}g_{\mathrm{SR}}}{\gamma_{\mathrm{R}}\!\!+\!\gamma_{\mathrm{D}}} , g_{\mathrm{SR}}>\dfrac{(\phi\!-\!1) (\gamma_{\mathrm{D}}\!+\!\gamma_{\mathrm{R}})}{\gamma_{\mathrm{R}} \gamma_{\mathrm{S}}},\gamma_{\mathrm{S}}^2 \gamma_{\mathrm{R}} \gamma_{\mathrm{D}} g_{\mathrm{SR}}^2 \\
\nonumber &+\! \gamma_{\mathrm{S}} \left(\gamma_{\mathrm{R}} \!\!+\!\! \gamma_{\mathrm{D}}\right) \left(\left(1\!\!-\!\!2 \phi\right)\gamma_{\mathrm{D}}\!\!+\!\! \left(1\!\!-\!\!\phi\right) \gamma_{\mathrm{R}} \right)\! g_{\mathrm{SR}}\!\!-\!\!\left(\phi\!\!-\!\!1\right) \left(\gamma_{\mathrm{R}}\!\! +\!\! \gamma_{\mathrm{D}}\right)^2\!\! <\!\!0 \bigg)\\
\nonumber &= \underbrace{ \Pr \left( g_{\mathrm{SR}} < \dfrac{\left(\gamma_{\mathrm{R}} + \gamma_{\mathrm{D}}\right) \left(\phi-1\right) }{\gamma_{\mathrm{S}}\gamma_{\mathrm{R}}}, g_{\mathrm{RD}} > \dfrac{\gamma_{\mathrm{S}}g_{\mathrm{SR}}}{\gamma_{\mathrm{R}}+\gamma_{\mathrm{D}}}   \right)}_{T_1}\\
\nonumber&+ \Pr\bigg(\underbrace{\dfrac{(\phi -1) (\gamma_{\mathrm{D}}+\gamma_{\mathrm{R}})}{\gamma_{\mathrm{R}} \gamma_{\mathrm{S}}}<g_{\mathrm{SR}}<\dfrac{(\gamma_{\mathrm{D}}+\gamma_{\mathrm{R}}) \varphi_1}{2 \gamma_{\mathrm{D}} \gamma_{\mathrm{R}} \gamma_{\mathrm{S}}},}_{T_{2,1}} \\
& \underbrace{ \frac{\gamma_{\mathrm{S}} g_{\mathrm{SR}}}{\gamma_{\mathrm{D}}\!\!+\!\!\gamma_{\mathrm{R}}}\!\!\leq\!\! g_{\mathrm{RD}}\!\!<\!\!\frac{\gamma_{\mathrm{D}}\!\!+\!\!\gamma_{\mathrm{R}}\!\!-\!\!\phi  (\gamma_{\mathrm{D}}\!\!+\!\!\gamma_{\mathrm{R}}) (\gamma_{\mathrm{S}} g_{\mathrm{SR}}\!\!+\!\!1)\!\!+\!\!\gamma_{\mathrm{R}} \gamma_{\mathrm{S}} g_{\mathrm{SR}}}{\gamma_{\mathrm{D}} \phi  (\gamma_{\mathrm{D}}\!\!+\!\!\gamma_{\mathrm{R}})\!\!-\!\!\gamma_{\mathrm{D}} (\gamma_{\mathrm{D}}\!\!+\!\!\gamma_{\mathrm{R}}\!\!+\!\!\gamma_{\mathrm{R}} \gamma_{\mathrm{S}} g_{\mathrm{SR}})}}_{T_{2,2}}\bigg) ,
\end{align}
with { $ \varphi_1$$=$$\sqrt{\!-2 \phi  \left(2 \gamma_{\mathrm{D}}^2\!+\!\gamma_{\mathrm{D}} \gamma_{\mathrm{R}}\!+\!\gamma_{\mathrm{R}}^2\right)\!+\!\phi ^2 (2 \gamma_{\mathrm{D}}\!+\!\gamma_{\mathrm{R}})^2\!+\!(\gamma_{\mathrm{D}}\!-\!\gamma_{\mathrm{R}})^2}\!+\!\gamma_{\mathrm{D}} (2 \phi\! -\!1)+\gamma_{\mathrm{R}} (\phi \!-\!1)$}, regarding that $g_{\mathrm{SR}}$ and $g_{\mathrm{RD}}$ are exponentially distributed random variables, the probabilities $T_1$ and $T_2$$=$$\Pr \left(T_{2,1},T_{2,2}\right)$ can be obtained as

\begin{align}
\nonumber T_1=& \int_{0}^{\dfrac{\left(\gamma_{\mathrm{R}} + \gamma_{\mathrm{D}}\right) \left(\phi-1\right) }{\gamma_{\mathrm{S}}\gamma_{\mathrm{R}}}} F_{g_{\mathrm{RD}}}\left(\dfrac{\gamma_{\mathrm{S}}}{\gamma_{\mathrm{R}}+\gamma_{\mathrm{D}}} g_{\mathrm{SR}}\right) f_{g_{\mathrm{SR}}}\left(x\right) dx\\
=& -\frac{\Omega_{\mathrm{RD}} (\eta_2+\eta_3) \left(\exp \left(-\frac{(\phi -1) (\eta_1 \Omega_{\textrm{SR}}+\Omega_{\mathrm{RD}} (\eta_2+\eta_3))}{\gamma_P \eta_1 \eta_2 \Omega_{\mathrm{RD}} \Omega_{\textrm{SR}}}\right)-1\right)}{\eta_1 \Omega_{\textrm{SR}}+\Omega_{\mathrm{RD}} (\eta_2+\eta_3)} \label{eq:T1fin} \\
\nonumber T_2=& \int_{\psi_4}^{\psi_5}\!\! F_{g_{\mathrm{RD}}}\left(\frac{\gamma_{\mathrm{D}}+\gamma_{\mathrm{R}}-\phi  (\gamma_{\mathrm{D}}+\gamma_{\mathrm{R}}) (\gamma_{\mathrm{S}} g_{\mathrm{SR}}+1)+\gamma_{\mathrm{R}} \gamma_{\mathrm{S}} g_{\mathrm{SR}}}{\gamma_{\mathrm{D}} \phi  (\gamma_{\mathrm{D}}+\gamma_{\mathrm{R}})-\gamma_{\mathrm{D}} (\gamma_{\mathrm{D}}+\gamma_{\mathrm{R}}+\gamma_{\mathrm{R}} \gamma_{\mathrm{S}} g_{\mathrm{SR}})}\right)\\
\nonumber & \times f_{g_{\mathrm{SR}}}\left(x\right) dx- \int_{\psi_4}^{\psi_5}F_{g_{\mathrm{RD}}}\left(\frac{\gamma_{\mathrm{S}} g_{\mathrm{SR}}}{\gamma_{\mathrm{D}}+\gamma_{\mathrm{R}}}\right) f_{g_{\mathrm{SR}}}\left(x\right) dx\\
\nonumber=&  \int_{\psi_4}^{\psi_5}\!\!\!\!\! F_{g_{\mathrm{RD}}}\left(\frac{\gamma_{\mathrm{D}}\!\!+\!\!\gamma_{\mathrm{R}}\!\!-\!\!\phi  (\gamma_{\mathrm{D}}\!\!+\!\!\gamma_{\mathrm{R}}) (\gamma_{\mathrm{S}} g_{\mathrm{SR}}\!\!+\!\!1)\!\!+\!\!\gamma_{\mathrm{R}} \gamma_{\mathrm{S}} g_{\mathrm{SR}}}{\gamma_{\mathrm{D}} \phi  (\gamma_{\mathrm{D}}\!\!+\!\!\gamma_{\mathrm{R}})\!\!-\!\!\gamma_{\mathrm{D}} (\gamma_{\mathrm{D}}\!\!+\!\!\gamma_{\mathrm{R}}\!\!+\!\!\gamma_{\mathrm{R}} \gamma_{\mathrm{S}} g_{\mathrm{SR}})}\right) \! f_{g_{\mathrm{SR}}}\!\left(x\right)\! dx\\
\nonumber &+\!\! \dfrac{\Omega_{\textrm{RD}} (\eta_2\!\!+\!\!\eta_3)}{\eta_1 \Omega_{\mathrm{SR}}\!\!+\!\!\Omega_{\textrm{RD}} (\eta_2\!\!+\!\!\eta_3)}\!\left(\!-\exp\left[\!{-\frac{\psi_5 (\eta_1 \Omega_{\mathrm{SR}}\!\!+\!\!\Omega_{\textrm{RD}} (\eta_2\!\!+\!\!\eta_3))}{\Omega_{\textrm{RD}}  (\eta_2\!\!+\!\!\eta_3)}}\right] \right.\\
& \left. +\exp\left[{-\dfrac{(\phi -1) (\eta_1 \Omega_{\mathrm{SR}}+\Omega_{\textrm{RD}} (\eta_2+\eta_3))}{\gamma_P \eta_1 \eta_2 \Omega_{\textrm{RD}} \Omega_{\mathrm{SR}}}}\right]\right)-e^{-\psi_4}+e^{-\psi_5}. \label{eq:T2fin}
\end{align}

\begin{figure}[h]\centering
	\includegraphics[width=0.5\linewidth]{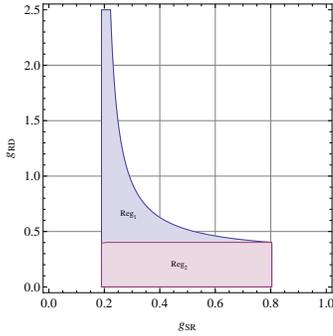}
	\caption{Regions for integral in~\eqref{eq:T2fin}, for $\gamma_P=15$ dB, $\phi=2$, and EPA.}
	\label{fig:integralregions}
\end{figure}

A closed-form approximation for the integral in~\eqref{eq:T2fin}, which is accurate from medium-to-high SNR regime, can be  obtained by approximating the corresponding integral region $\mathrm{Reg}_1$ (almost null at high SNR), which is formed by considering the events $T_{2,1}$ and $T_{2,2}^{r}= g_{\mathrm{RD}}<\frac{\gamma_{\mathrm{D}}+\gamma_{\mathrm{R}}-\phi  (\gamma_{\mathrm{D}}+\gamma_{\mathrm{R}}) (\gamma_{\mathrm{S}} g_{\mathrm{SR}}+1)+\gamma_{\mathrm{R}} \gamma_{\mathrm{S}} g_{\mathrm{SR}}}{\gamma_{\mathrm{D}} \phi  (\gamma_{\mathrm{D}}+\gamma_{\mathrm{R}})-\gamma_{\mathrm{D}} (\gamma_{\mathrm{D}}+\gamma_{\mathrm{R}}+\gamma_{\mathrm{R}} \gamma_{\mathrm{S}} g_{\mathrm{SR}})}$, by an approximate region $\mathrm{Reg}_2$, as depicted in Fig.~\ref{fig:integralregions}, which considers the event $T_{2,2}^{\mathrm{ap}}=g_{\mathrm{RD}}<\dfrac{\varphi_1}{2 \gamma_{\mathrm{D}} \gamma_{\mathrm{R}} }$ instead of the event $T_{2,2}^{r}$. Then, $T_2$ in~\eqref{eq:T2fin} can be approximated by
\begin{align}\label{eq:T2approx}
\nonumber T_2 \approx \tilde{T_2}=& \left(e^{-\psi_5}-e^{-\psi_4}\right)  e^{-\frac{\eta_1\Omega_{\mathrm{SR}}\psi_5}{\Omega_{\mathrm{RD}} (\eta_2+\eta_3)}} + \dfrac{\Omega_{\textrm{RD}} (\eta_2+\eta_3)}{\eta_1 \Omega_{\mathrm{SR}}+\Omega_{\textrm{RD}} (\eta_2+\eta_3)}\\
\nonumber & \times \left(\exp\left[{-\dfrac{(\phi -1) (\eta_1 \Omega_{\mathrm{SR}}+\Omega_{\textrm{RD}} (\eta_2+\eta_3))}{\gamma_P \eta_1 \eta_2 \Omega_{\textrm{RD}} \Omega_{\mathrm{SR}}}}\right] \right.\\
&- \left. \exp\left[{-\frac{\psi_5 (\eta_1 \Omega_{\mathrm{SR}}+\Omega_{\textrm{RD}} (\eta_2+\eta_3))}{\Omega_{\textrm{RD}} (\eta_2+\eta_3)}}\right]\right).
\end{align}

By following the same reasoning, $J_2$ can be also splitted into two probability terms as
\begin{align}\label{eq:J2}
J_2=& \underbrace{\Pr \left( g_{\mathrm{SR}}<\varphi_2, g_{\mathrm{RD}}<\dfrac{\gamma_{\mathrm{S}}g_{\mathrm{SR}}}{\gamma_{\mathrm{R}}+\gamma_{\mathrm{D}}} \right)}_{T_3}+\underbrace{\Pr \bigg( g_{\mathrm{SR}}>\varphi_2,g_{\mathrm{RD}}<\varphi_3 \bigg)}_{T_4},
\end{align}
with
\begin{align}\label{eq:varphi2}
\nonumber &\varphi_2=\dfrac{\phi  \left(2 \gamma_{\mathrm{D}}^2+3 \gamma_{\mathrm{D}} \gamma_{\mathrm{R}}+\gamma_{\mathrm{R}}^2\right)-(\gamma_{\mathrm{D}}+\gamma_{\mathrm{R}})^2}{2 \gamma_{\mathrm{D}} \gamma_{\mathrm{R}} \gamma_{\mathrm{S}}}\\
&\!\!+\!\!\frac{1}{2}\!\! \sqrt{\dfrac{\!\!\phi^2\!\! \left(2 \gamma_{\mathrm{D}}^2\!\!+\!\!3 \gamma_{\mathrm{D}} \gamma_{\mathrm{R}}\!\!+\!\!\gamma_{\mathrm{R}}^2\right)^2\!\!\!-\!\!2 \phi\!\!  \left(2 \gamma_{\mathrm{D}}^2\!\!+\!\!\gamma_{\mathrm{D}} \gamma_{\mathrm{R}}\!\!+\!\!\gamma_{\mathrm{R}}^2\right)\!\! (\gamma_{\mathrm{D}}\!\!+\!\!\gamma_{\mathrm{R}})^2\!\!+\!\!\left(\gamma_{\mathrm{D}}^2\!\!-\!\!\gamma_{\mathrm{R}}^2\right)^2}{\gamma_{\mathrm{D}}^2 \gamma_{\mathrm{R}}^2 \gamma_{\mathrm{S}}^2}},
\end{align}
\begin{align}\label{eq:varphi3}
\nonumber \varphi_3=&\frac{\gamma_{\mathrm{D}} \phi -\gamma_{\mathrm{D}}-\gamma_{\mathrm{R}}}{2 \gamma_{\mathrm{D}} \gamma_{\mathrm{R}}}\\
&+\frac{1}{2} \sqrt{\frac{\gamma_{\mathrm{D}}^2 \phi ^2+(\gamma_{\mathrm{D}}-\gamma_{\mathrm{R}})^2-2 \gamma_{\mathrm{D}} \phi  (\gamma_{\mathrm{D}}-\gamma_{\mathrm{R}}-2 \gamma_{\mathrm{R}} \gamma_{\mathrm{S}} x)}{\gamma_{\mathrm{D}}^2 \gamma_{\mathrm{R}}^2}}
\end{align}
where $T_3$ and $T_4$ can be obtained as
\begin{align}\label{eq:T3fin}
\nonumber T_3=& \int_{0}^{\varphi_2} F_{g_{\mathrm{RD}}}\left(\dfrac{\gamma_{\mathrm{S}}}{\gamma_{\mathrm{R}}+\gamma_{\mathrm{D}}} g_{\mathrm{SR}}\right) f_{g_{\mathrm{SR}}}\left(x\right) dx\\
\nonumber = & \frac{\Omega_{\mathrm{RD}} (\eta_2+\eta_3) \left(\exp\left[-\frac{\eta_1 \psi_5 \Omega_{\mathrm{SR}}+\eta_2 \psi_5 \Omega_{\mathrm{RD}}+\eta_3 \psi_5 \Omega_{\mathrm{RD}}}{\eta_2 \Omega_{\mathrm{RD}} \Omega_{\mathrm{SR}}+\eta_3 \Omega_{\mathrm{RD}} \Omega_{\mathrm{SR}}}\right]-1\right)}{\eta_1 \Omega_{\mathrm{SR}}+\eta_2 \Omega_{\mathrm{RD}}+\eta_3 \Omega_{\mathrm{RD}}}\\
&-e^{-\frac{\psi_5}{\Omega_{\mathrm{SR}}}}+1, \end{align}
\begin{align}\label{eq:T4fin}
\nonumber T_4=& \int_{\varphi_2}^{\infty} F_{g_{\mathrm{RD}}}\left(\varphi_3\right) f_{g_{\mathrm{SR}}}\left(x\right) dx \\
\nonumber	=& \frac{\sqrt{\pi } \eta_1 \phi  \Omega_{\mathrm{SR}}\exp \left(\psi_6\right)}{2 \Omega_{\mathrm{RD}} \sqrt{\gamma_P \eta_1 \eta_2 \eta_3 \phi  \Omega_{\mathrm{SR}}}}   \mathrm{erfc}\left(\frac{\eta_1 \phi  \Omega_{\mathrm{SR}}+\Omega_{\mathrm{RD}} \sqrt{\varphi_4}}{2 \Omega_{\mathrm{RD}} \sqrt{\gamma_P \eta_1 \eta_2 \eta_3 \phi  \Omega_{\mathrm{SR}}}}\right)\\
&+e^{-\frac{\varphi_2}{\Omega_{\mathrm{SR}}}} \left(1-\exp \left(\frac{\eta_2-\eta_3 \phi +\eta_3-\sqrt{\varphi_4}}{2 \gamma_P \eta_2 \eta_3 \Omega_{\mathrm{RD}}}\right)\right),
\end{align}
\noindent where $T_4$ can be solved by using~\cite[Eq.~(3.322-1)]{book:gradshteyn},  $\varphi_4=\eta_2^2+\eta_3^2 (\phi -1)^2+2 \eta_2 \eta_3 (2 \gamma_P \eta_1 \varphi_2 \phi +\phi -1)$ and $\psi_6$ is defined in~\ref{eq:gsopfinal}. Finally, by adding the terms $J_1 \approx T_1+\tilde{T_2}$ and $J_2=T_3+T_4$ and making some simplifications,~\eqref{eq:gsopfinal} can be obtained.
\section{Proof of Proposition \ref{prop:afe}}
\label{app:D}

From~\eqref{eq:Deltau_1br}, the average fractional equivocation can be formulated as
\begin{align}\label{eq:afe}
\nonumber	\bar \Delta =& \mathbb{E}\{\Delta\}\\
\nonumber	{\approx}& \int_{2^{2 R_S}}^{\infty} f_{\Phi}\left(\phi\right) d\phi+\int_{1}^{2^{2 R_S}} \dfrac{1}{{2 R_S}} \log_2 \phi f_{\Phi}\left(\phi\right) d\phi\\
\nonumber \stackrel{(b)}{=}&1-F_{\Phi}\left(2^{2 R_S}\right)+\dfrac{1}{\ln(2) {2 R_S}}\left[\ln(\phi) F_{\Phi}\left(\Phi\right) \bigg|_1^{2^{2 R_S}} \right.\\
& \left.-\int_{1}^{2^{2 R_S}}\dfrac{1}{\phi}F_{\Phi}\left(\phi\right)d \phi\right] 	,
\end{align}
where $(b)$ is obtained by applying integration by parts, i.e. $\int_{a}^{b} u(t) v' (t)dt= u(b)v(b)-u(a)v(a)-\int_{a}^{b}u'(t)v(t)dt$ with $u= \ln(\phi)$  and $dv= f_{\Phi}(\phi)$. Then, after the corresponding substitutions and some simplifications, we arrive to~\eqref{eq:afep}.

\section{Proof of Proposition \ref{prop:throughput}}
\label{app:B}
The throughput of confidential transmission can be defined in terms of the probability of successful transmission and the target secrecy rate, and it can be expressed as
\begin{align}\label{eq:sectroughproof}
\nonumber \mathcal{T}=& \mathbb{P}_{\mathrm{ST}} R_S\\
\nonumber =& \Pr\left(\Gamma_L \geq 2^{2 R_T}-1 \stackrel{\Delta}{=} \tau_2\right) R_S\\
\nonumber \stackrel{(b)}{\approx} & \Pr \left( \dfrac{\gamma_{\mathrm{R}}}{\gamma_{\mathrm{R}}+\gamma_{\mathrm{D}}} \min\{\gamma_{\mathrm{S}} g_{\mathrm{SR}},\left(\gamma_{\mathrm{R}}+\gamma_{\mathrm{D}}\right) g_{\mathrm{RD}}\}\geq \tau_2 \right) R_S \\
\nonumber =& \Pr \left( g_{\mathrm{SR}} \geq  \dfrac{\tau_2\left(\gamma_{\mathrm{R}}+\gamma_{\mathrm{D}}\right)}{\gamma_{\mathrm{R}} \gamma_{\mathrm{S}} }, g_{\mathrm{RD}}\geq \dfrac{\tau_2 }{\gamma_{\mathrm{R}}}\right) R_S \\
=&R_S \left[1-F_{g_{\mathrm{SR}}}\left(\dfrac{\tau_2\left(\gamma_{\mathrm{R}}+\gamma_{\mathrm{D}}\right)}{\gamma_{\mathrm{R}} \gamma_{\mathrm{S}} }\right)\right] \left[1-F_{g_{\mathrm{RD}}}\left(\dfrac{\tau_2 }{\gamma_{\mathrm{R}}}\right)\right] ,
\end{align}
where $(b)$ was obtained by considering the approximation in~\eqref{eq:snrDapprox}. Then, by considering that $g_{\mathrm{SR}}$ and $g_{\mathrm{RD}}$ are exponentially distributed random variables and after some simplifications, $\mathcal{T}$ is obtained as in~\eqref{eq:throughput}.
\section{Proof of Proposition~\ref{cor:thrfeas}}
\label{app:C}
The maximum achievable throughput of the confidential transmission can be obtained by solving the following optimization problem 
\begin{align*}
\max_{\mathrm{w.r.t.}\ R_S, R_T, \eta_1, \eta_2, \eta_3} \mathcal{T}
\end{align*}
\vspace{-5mm}
\begin{align}\label{eq:maxtrh}
\mathrm{s.t.}\ R_T\!\geq\! R_S>0, \eta_1\!>\!0, \eta_2\!>\!0, \eta_3\!>\!0, \eta_1+\eta_2+\eta_3\!=\!1.
\end{align}
From~\eqref{eq:throughput}, it can be noticed that $\mathcal{T}$ is a monotonically decreasing function over $R_T$, in its valid interval, thus we can consider that the value of $R_T$ that maximizes $\mathcal{T}$ is given by the minimum permitted value, i.e., $R_T=R_S$. Moreover, $\eta_1$ in~\eqref{eq:throughput} can be expressed as $\eta_1=1-\eta_2-\eta_3$, and,
 $\mathcal{T}$ is concave w.r.t. $\eta_2$, while it is monotonically decreasing w.r.t. $\eta_3$. Therefore, the optimization problem in~\eqref{eq:maxtrh} can be re-expressed as
\begin{align*}
\max_{\mathrm{w.r.t.}\ R_S, \eta_2} \mathcal{T'} =\mathcal{T} \mid_{R_T=R_S, \eta_1=1-\eta_2-\eta_3}
\end{align*}
\vspace{-5mm}
\begin{align}\label{eq:maxtrhsimp}
\mathrm{s.t.}\ R_S>0,  \eta_2>0, \eta_3>0, \eta_2+\eta_3<1.
\end{align}
Hence, the optimal values for $R_S^*$ and $\eta_2^*$ that maximize the throughput are obtained by solving the following equations
\begin{align}\label{eq:partderivrs}
\nonumber \dfrac{\partial  \mathcal{T'}}{\partial R_S}=&\frac{\exp \left(\frac{\left(2^{2 R_S}-1\right) ((\eta_2+\eta_3) (\Omega_{\mathrm{RD}}-\Omega_{\mathrm{SR}})+\Omega_{\mathrm{SR}})}{\gamma_P \eta_2 \Omega_{\mathrm{RD}} \Omega_{\mathrm{SR}} (\eta_2+\eta_3-1)}\right)}{\gamma_P \eta_2 \Omega_{\mathrm{RD}} \Omega_{\mathrm{SR}} (\eta_2+\eta_3-1)}\\
\nonumber &\times \left[\gamma_P \eta_2 \Omega_{\mathrm{RD}} \Omega_{\mathrm{SR}} (\eta_2+\eta_3-1)+2^{2 R_S+1} R_S \ln(2) \right.\\
\nonumber&\times \left. ((\eta_2+\eta_3) (\Omega_{\mathrm{RD}}-\Omega_{\mathrm{SR}})+\Omega_{\mathrm{SR}})\right]=0\\
2^{2 R_S+1} R_S=&\dfrac{\gamma_P \eta_2 \Omega_{\mathrm{RD}} \Omega_{\mathrm{SR}} (\eta_2+\eta_3-1)}{\ln(2) ((\eta_2+\eta_3) (\Omega_{\mathrm{RD}}-\Omega_{\mathrm{SR}})+\Omega_{\mathrm{SR}})},
\end{align}
\begin{align}\label{eq:partderiveta2}
\nonumber \dfrac{\partial  \mathcal{T'}}{\partial \eta_2}=&-\frac{\exp \left(\frac{\left(2^{ 2 R_S}-1\right) ((\eta_2+\eta_3) (\Omega_{\mathrm{RD}}-\Omega_{\mathrm{SR}})+\Omega_{\mathrm{SR}})}{\gamma_P\eta_2 \Omega_{\mathrm{RD}} \Omega_{\mathrm{SR}} (\eta_2+\eta_3-1)}\right)}{\gamma_P\eta_2^2 \Omega_{\mathrm{RD}} \Omega_{\mathrm{SR}} (\eta_2+\eta_3-1)^2}\\
\nonumber&\times \left(2^{2 R_S}-1\right) R_S \left[\Omega_{\mathrm{RD}} \left((\eta_2+\eta_3)^2-\eta_3\right) \right.\\
\nonumber &\left. -\Omega_{\mathrm{SR}} (\eta_2+\eta_3-1)^2\right]=0\\
\nonumber 0=&\Omega_{\mathrm{RD}} \left((\eta_2+\eta_3)^2-\eta_3\right)\Omega_{\mathrm{SR}} (\eta_2+\eta_3-1)^2\\
\nonumber 0=&\eta_2^2 (\Omega_{\mathrm{RD}}-\Omega_{\mathrm{SR}})+2 \eta_2 (\eta_3 (\Omega_{\mathrm{RD}}-\Omega_{\mathrm{SR}})+\Omega_{\mathrm{SR}})\\
&+(\eta_3-1) (\eta_3 (\Omega_{\mathrm{RD}}-\Omega_{\mathrm{SR}})+\Omega_{\mathrm{SR}}).
\end{align}
Thus, after solving these equations and some simplifications, the expressions for $R_S^*$ and $\eta_2^*$ in~\eqref{eq:Rsmaxt} and~\eqref{eq:etamax} are obtained, which are the solutions that satisfy the constraints in~\eqref{eq:maxtrhsimp}.

\end{document}